\documentclass[times,sort&compress,3p]{elsarticle}
\usepackage[labelfont=bf]{caption}
\journal{Journal of Multivariate Analysis}
\usepackage{amssymb,geometry,graphicx,setspace, scalefnt, verbatim, amsfonts,bm,float,booktabs,multirow,amsmath,mathtools,color,dsfont,arydshln,subcaption,amsthm}
\usepackage[left]{lineno}
\usepackage{enumerate}
\usepackage{dsfont}
\usepackage{cancel}
\usepackage[normalem]{ulem}
\usepackage{xcolor}
\usepackage{chngcntr}
\usepackage[framemethod=TikZ]{mdframed}
\usepackage{tikz}
\usetikzlibrary{shadows}
\usepackage{xcolor}

\newenvironment{reminderbox}{
  \begin{mdframed}[
    backgroundcolor=gray!15,
    linecolor=gray!60,
    shadow=true,
    shadowsize=4pt,
    roundcorner=5pt,
    innertopmargin=1em,
    innerbottommargin=1em,
    innerleftmargin=1em,
    innerrightmargin=1em
  ]
  \small
}{
  \end{mdframed}
}

\usepackage[ruled,vlined,linesnumbered]{algorithm2e}
\usepackage[
    colorlinks=true,
    linkcolor=blue,
    citecolor=blue,
    filecolor=magenta,
    urlcolor=blue,
    pdfborder={0 0 0},
]{hyperref}
\usepackage{orcidlink}
\usepackage{longtable}

\newcommand{\dd}{\mathop{}\!\mathrm{d}}
\allowdisplaybreaks
\newcommand{\SetAlgoItemize}{\setlength{\leftmargini}{1em}
}

\newtheorem{theorem}{Theorem}\newtheorem{definition}{Definition}\newtheorem{proposition}{Proposition}\newtheorem{remark}{Remark}

\makeatletter
\patchcmd{\abstract}{\vskip\abs@vskip}{\vskip0pt}{}{}
\let\@date\relax
\patchcmd{\maketitle}{\par\vskip 0em \relax}{\par\vskip 0em \relax}{}{}

\long\def\pprintMaketitle{\clearpage
  \iflongmktitle\if@twocolumn\let\columnwidth=\textwidth\fi\fi
  \resetTitleCounters
  \def\baselinestretch{1}\printFirstPageNotes
  \begin{\elsarticletitlealign}\thispagestyle{pprintTitle}\def\baselinestretch{1}\Large\@title\par\vskip6pt\ifx\@elsarticlenewpageafter\newpage@after@title \newpage
    \fi \ifdoubleblind
      \vspace*{2pc}
    \else
      \normalsize\elsauthors\par\vskip10pt
      \footnotesize\itshape\elsaddress\par\vskip1em \fi
    \ifx\@elsarticlenewpageafter\newpage@after@author \newpage
    \fi \hrule\vskip12pt
    \ifvoid\absbox\else\unvbox\absbox\par\vskip10pt\fi
    \ifvoid\keybox\else\unvbox\keybox\par\vskip10pt\fi
    \hrule\vskip12pt
    \ifx\@elsarticlenewpageafter\newpage@after@abstract \newpage
    \fi \end{\elsarticletitlealign}\gdef\thefootnote{\arabic{footnote}}}
  \patchcmd{\pprintMaketitle}{\clearpage}{}{}{}
\makeatother

\AtBeginEnvironment{algorithm}{\setstretch{1}}

\begin{document}

\vspace*{-6em}
\begin{center}
\begin{reminderbox}
  \centering
  \textbf{\textsf{This is a preprint. The revised version of this paper is published as}}\\\smallskip
  Andr{\'{e}} F. B. Menezes, Andrew C. Parnell, and Keefe Murphy (2025),
  ``Finite mixture representations of zero-and-$N$-inflated distributions for count-compositional data'',
  \textit{Journal of Multivariate Analysis}, 210:105492.\\ \href{https://doi.org/10.1016/j.jmva.2025.105492}{doi: 10.1016/j.jmva.2025.105492}
\end{reminderbox}
\end{center}

\begin{frontmatter}

\author[1]{Andr{\'e} F. B. Menezes\,\orcidlink{0000-0002-3320-9834}\corref{mycorrespondingauthor}}
\author[2]{Andrew C. Parnell\,\orcidlink{0000-0001-7956-7939}}
\author[1]{Keefe Murphy\,\orcidlink{0000-0002-7709-3159}}

\address[1]{Hamilton Institute and Department of Mathematics and Statistics, Maynooth University, Ireland}
\address[2]{School of Mathematics and Statistics, Insight Centre for Data Analytics, University College Dublin, Ireland}

\title{\texorpdfstring{Finite mixture representations of zero-and-$N$-inflated distributions for count-compositional data}{Finite mixture representations of zero-and-N-inflated distributions for count-compositional data}}

\cortext[mycorrespondingauthor]{Corresponding author. Email address: \url{andrefelipemaringa@gmail.com} (A.F.B. Menezes).}

\begin{abstract}
We provide novel probabilistic portrayals of two multivariate models designed to handle zero-inflation in count-compositional data.
We develop a new unifying framework that represents both as finite mixture distributions.
One of these distributions, based on Dirichlet-multinomial components, has been studied before, but has not yet been properly characterised as a sampling distribution of the counts.
The other, based on multinomial components, is a new contribution.
Using our finite mixture representations enables us to derive key statistical properties, including moments, marginal distributions, and special cases for both distributions.
We develop enhanced Bayesian inference schemes with efficient Gibbs sampling updates, wherever possible, for parameters and auxiliary variables, demonstrating improvements over existing methods in the literature.
We conduct simulation studies to evaluate the efficiency of the Bayesian inference procedures and
present applications to a human gut microbiome dataset
to illustrate the practical utility of the proposed distributions.
\end{abstract}


\end{frontmatter}

\section{Introduction}\label{sec:introduction}

The excess of zeros in count-compositional data occurs when one or more categories have a larger number of observed
zeros than expected under common statistical distributions, such as the multinomial and Dirichlet-multinomial.
The primary complication in multivariate count-compositional settings, compared to univariate cases,
is that the excess zeros can occur in a single category or across multiple categories. In extreme cases where zeros co-occur in all but one category, the count for the remaining category will coincide with the number of trials, $N$. This phenomenon has been referred to in the univariate zero-inflation literature as `endpoint-inflation' \citep{Tian2015,Dupuy2017}, in relation to a zero-inflated extension of the binomial distribution first proposed by \citet{Deng2015} However, we stress throughout this paper that zero-inflation in multivariate settings can also induce another type of inflation. Specifically, when zeros co-occur while more than one category is non-zero, the counts in the non-zero categories will also inflate, as the excess zeros in some categories effectively redistribute probability mass to the others.

We discuss two multivariate probability distributions designed to address the prevalence of excess zeros
when modelling count-compositional data. We introduce a unifying framework that represents these distributions as finite mixtures.
Specifically, we derive a novel zero-and-$N$-inflated multinomial (ZANIM) distribution,
which is based on multinomial mixture components, and
then show that our framework also incorporates
the zero-and-$N$-inflated Dirichlet-multinomial (ZANIDM) distribution.
Although ZANIDM was first introduced by \citet{Koslovsky2023}, under the name ZIDM, it was described only through a stochastic representation via a mixture distribution on the count probabilities. We fully characterise ZANIDM as a sampling distribution on the counts capable of simultaneously modelling both zero-and-$N$-inflation and overdispersion using Dirichlet-multinomial mixture components.

Our paper is structured as follows. We derive the finite mixture representations of both distributions in Section \ref{sec:model_derivation}.
These representations facilitate the derivation of some key theoretical properties of the distributions in Section \ref{sec:properties}. We propose Bayesian frameworks for model inference in Section \ref{sec:inference}. In particular, for ZANIDM, we note that it is possible to improve the efficiency of the MCMC algorithm by
marginalising out a latent variable, which was not done by \citet{Koslovsky2023}. We present two simulation studies in Section \ref{sec:simulation_study} which i) compare different MCMC
algorithms for inferring the parameters of ZANIDM and ii) illustrate the practical utility of both distributions
when dealing with zero-inflation in count-compositional data. Applications to a human gut microbiome dataset follow in Section \ref{sec:microbiome}. We conclude with a brief discussion in Section \ref{sec:discussion}. Additional details on associated derivations and inference schemes for both distributions are deferred to the Appendices. We also provide further results in the Supplementary Material. For now, we begin by describing some related proposals.

\subsection{Related work}
\label{sec:lit_review}

Many extensions of the multinomial distribution have been proposed in the literature,
most of which aim to address extra variation relative to its inherent limitations,
particularly its negative covariance structure. The
Dirichlet-multinomial \citep{Mosimann1962}, finite mixture of multinomials
\citep{Morel1993}, and Conway-Maxwell-multinomial \citep{Kadane2018,Morris2020} distributions are notable examples.
In contrast, there are relatively few extensions of the multinomial or other
count-compositional distributions which address the issue of excess zeros.
Key contributions in this area include:
\citet{Diallo2018}, who studied a specific case of zero-inflation in the multinomial distribution; \citet{Tang2019}, who
introduced the zero-inflated generalised Dirichlet-multinomial by modifying the beta stick-breaking representation
of the generalised Dirichlet distribution \citep{Connor1969} to incorporate a zero-augmented beta distribution; and
\citet{Tuyl2019}, who proposed a spike-and-slab prior for the multinomial probability parameter,
assigning positive probability mass to zero.

More recently, \citet{Zeng2023} proposed the zero-inflated logistic normal multinomial,
while \citet{Koslovsky2023} introduced a zero-inflated extension of the Dirichlet-multinomial,
leveraging its gamma representation to incorporate zero-inflation into the count probabilities.
Notably, both of these recent works focus on modifying the latent space rather than the sampling
distribution of the counts, and we note some similarities in their derivations. However, \citet{Zeng2023} and \citet{Koslovsky2023} primarily focus on modelling data from human microbiome studies, without
emphasising the theoretical properties of their models.
Our paper addresses this gap by: (i) presenting a unified framework for both models,
reformulating them in terms of their unconditional representation without latent variables,
(ii) leveraging this framework to derive key statistical properties for both models, and
(iii) discussing Bayesian inference using their latent structures while proposing improvements to the sampling scheme for the ZANIDM distribution. \section{Derivation of the distributions}\label{sec:model_derivation}
In what follows, we shall assume that
$\mathbf{Y} = (Y_{1}, \ldots, Y_{d})$ denotes a $d$-dimensional random vector of
count compositions, where $\mathbf{y} = (y_{1}, \ldots, y_{d})$ represents the observed
data.

\subsection{\texorpdfstring{Zero-and-$N$-inflated multinomial distribution}{Zero-and-N-inflated multinomial distribution}}\label{sec:zanim}

A well-known probability distribution for describing count-compositional data is
the multinomial distribution, whose probability mass function (PMF) is given by
\begin{equation}\label{eq:pmf_multinomial}
\Pr \lbrack \mathbf{Y} = \mathbf{y}; \bm{\theta} \rbrack =
\binom{N}{y_{1} \dots y_{d}}\prod_{j=1}^d\theta_j^{y_{j}}, \quad \mathbf{y} \in \Omega_{d, N},
\end{equation}
where the sample space $\bm{\Omega}_{d,N} = \left\{\mathbf{y} \in (0, \ldots, N)^d; \sum_{j=1}^d y_{j} = N\right\}$ is a $d$-dimensional discrete simplex and $\bm{\theta} = (\theta_1, \ldots, \theta_d)$
is a vector of category-specific success probability parameters, with each $\theta_j \geq 0$,
such that $\sum_{j=1}^d\theta_j=1$.
Here, $N$ is a known constant denoting the number of trials.
Within this paradigm, we shall consider the parameterisation
$\theta_j = \lambda_j/\sum_{k=1}^{d}\lambda_k,$
where $\lambda_j$ is a measure of the relative importance of category $j$;
when it is normalised, it gives the occurrence probability of category $j$, i.e., $\theta_j$.

We begin by expressing the multinomial likelihood as a product of
Poisson likelihoods. To obtain this formulation, we introduce the auxiliary random variable $\phi$, defined as
\begin{equation}\label{eq:gamma_trick}
(\phi \mid \bm{\lambda}, \mathbf{y}) \sim \operatorname{Gamma}\left\lbrack N,
\sum_{j=1}^{d}\lambda_j\right\rbrack,
\end{equation}
which leads to the following joint distribution for ($\mathbf{Y}$, $\phi$):
\begin{equation}\label{eq:pmf_multinomial_augmented}
p(\mathbf{y}, \phi; \bm{\lambda}) = \dfrac{N!\phi^{N - 1}}{\Gamma(N)}
\prod_{j=1}^d\left\lbrack\dfrac{\lambda_j^{y_{j}}e^{-\lambda_j\phi}}{y_{j}!}\right\rbrack,
\end{equation}
where the marginal distribution of $\mathbf{Y}$ is given by \eqref{eq:pmf_multinomial}.
We note that the augmented likelihood in \eqref{eq:pmf_multinomial_augmented} factors into independent terms for each $\lambda_j$ and that the same likelihood, up to a multiplicative constant, can be obtained via the
multinomial-Poisson transformation of \citet{Baker1994}.
For an arbitrary category $j$,
the likelihood contribution of an observation takes a `Poisson-type' form, i.e.,
\begin{equation}\label{eq:likelihood_contribution_category_j}
p(y_{j}, \phi \mid \lambda_j) \propto \dfrac{\lambda_j^{y_{j}}e^{-\lambda_j\phi}}{y_{j}!}, \quad y_{j} \in \{0, \ldots, N\}.
\end{equation}

The excess of zeros in multinomial count data may be structural in nature.
A reasonable way to address this using the augmented likelihood given in
\eqref{eq:pmf_multinomial_augmented} is to introduce additional parameters to account for
zero-inflation with respect to each category.
Specifically, we consider a mixture-type approach by adjusting the `Poisson-type' form in
\eqref{eq:likelihood_contribution_category_j} to include a zero-inflation parameter for each category, in the spirit of \citet{Lambert1992}.
Thus, the modified likelihood contribution for category $j$ adopts a `ZI-Poisson-type' form, i.e.,
\begin{equation}\label{eq:ZI_likelihood_contribution_category_j}
p(y_j, \phi \mid \lambda_j, \zeta_j) \propto \left(\zeta_j\mathds{1}_0(y_{j}) + (1 - \zeta_j)\dfrac{\lambda_j^{y_{j}}e^{-\lambda_j\phi}}{y_{j}!}\right),\quad y_{j} \in \{0, \ldots, N\},
\end{equation}
where $\zeta_j \in \lbrack 0, 1\rbrack$ denotes the probability of zero-inflation
of category $j$, which we henceforth refer to as the `excess-of-zero parameter',
and $\mathds{1}_0(y_{j})$ is the usual indicator function $\mathds{1}(y_{j}=0)$,
which evaluates to $1$ if $y_{j}=0$ or $0$ otherwise. It is important to note
that $Y_j$ may still be $0$ even if $\zeta_j=0$. We refer to such zeros as
`sampling zeros', in contrast to `structural zeros'.

Replacing the product of `Poisson-type' terms in \eqref{eq:pmf_multinomial_augmented}
with a product of independent `ZI-Poisson-type' terms from \eqref{eq:ZI_likelihood_contribution_category_j} yields
a new joint distribution for $(\mathbf{Y}, \phi)$ given by

\begin{equation}\label{eq:pmf_zanim_augmented}
p(\mathbf{y}, \phi; \bm{\lambda}, \bm{\zeta}) =
\dfrac{N!\phi^{N - 1}}{\Gamma(N)}\prod_{j=1}^d \left\lbrack
\zeta_j\mathds{1}_0(y_{j}) + (1 - \zeta_j)\dfrac{\lambda_j^{y_{j}}e^{-\lambda_j\phi}}{y_{j}!}
\right\rbrack.
\end{equation}
We now aim to marginalise out the latent variable $\phi$ from
\eqref{eq:pmf_zanim_augmented} and ensure that the function
$\int p(\mathbf{y}, \phi; \bm{\lambda}, \bm{\zeta}) \dd \phi$ will be a proper PMF.
We state our main result concerning the PMF of this new distribution,
named the zero-and-$N$-inflated multinomial (ZANIM) distribution, after first
introducing some notation that will help to represent it in the form of a finite mixture.

\begin{definition}\label{def:set_K}
Let $\mathfrak{K} = \{\mathcal{K} \subseteq \{1,\ldots,d\}; 1 \leq \lvert\mathcal{K}\rvert \leq d-2\}$
represent the set of all subsets $\mathcal{K}$ of $\{1,\ldots,d\}$ with cardinality $\lvert\mathcal{K}\rvert$ between $1$ and $d - 2$, such that the subsets $\mathcal{K}$ gather categories with counts of zero, excluding the cases where exactly $d$ and $d-1$ categories are zero-inflated.
\end{definition}

\begin{definition}\label{def:etas}
Let $\bm{\zeta} = (\zeta_1, \ldots, \zeta_d)$, with $\zeta_j \in \lbrack 0, 1\rbrack$, and let
$\eta_d = \prod_{j=1}^d(1-\zeta_j)$,
$\eta_0 = \prod_{j=1}^d\zeta_j$, and
$\eta_N^{(j)} = (1 - \zeta_{j})\prod_{k \colon k \neq j} \zeta_{k}$, for $j \in \{1,\ldots, d\}$.
We define
$\eta_{\mathcal{K}} = \prod_{k \in \mathcal{K}} \zeta_k\prod_{j \notin \mathcal{K}} (1 - \zeta_j)$ for each $\mathcal{K} \in \mathfrak{K}$ and let $\bm{\eta}_{\mathfrak{K}} = \{\eta_\mathcal{K}; \mathcal{K} \in \mathfrak{K}\}$ denote the set of such terms.
The full set of mixture weights
$\bm{\eta} = \{\eta_d, \eta_0, \eta_{N}^{(1)},\ldots,\eta_{N}^{(d)}, \bm{\eta}_{\mathfrak{K}}\}$ are functions of the $\bm{\zeta}$ parameters and sum to one as required.
\end{definition}
\begin{definition}\label{def:dirac}
Let $\delta_c(x)$ denote the Dirac measure with unit mass at $c$ for $x \in \mathbb{R}$ and denote the random variable with such a measure by $X \sim \delta_c(\cdot)$.
For $\mathbf{x} \in \mathbb{R}^d$, we define the multivariate Dirac measure as
$\delta_{\mathbf{0}_d}(\mathbf{x}) = \prod_{j=1}^d\delta_0(x_j)$ and denote the
corresponding random vector with unit mass at $\mathbf{0}_d$ by $\mathbf{X} \sim \delta_{\mathbf{0}_d}(\cdot)$.
\end{definition}
\begin{theorem}\label{theo:zanim_pmf}
The zero-and-$N$-inflated multinomial distribution, which we denote by
$\mathbf{Y} \sim \operatorname{ZANIM}_d\lbrack N, \bm{\theta}, \bm{\zeta} \rbrack$, is a finite mixture
distribution with $2^d$ components and PMF given by
\begin{align}
\label{eq:zanim_pmf__no_zero__mixture}
\Pr\lbrack
\mathbf{Y} = \mathbf{y}; \bm{\theta}, \bm{\zeta}\rbrack &= \eta_d\,\binom{N}{y_{1} \dots y_{d}}\,
\prod_{j=1}^d
\theta_j^{y_{j}} \\
\label{eq:zanim_pmf__only_one_no_zero__mixture}
&\phantom{=}~+
\sum_{j=1}^{d}\,
\eta_{N}^{(j)}
\left\lbrack
\mathds{1}_0\left(\sum_{k\colon k \neq j} y_{k} \right)\,
\right\rbrack \\
\label{eq:zanim_pmf__sets_of_zi_non_zi__mixture}
&\phantom{=}~+
\sum_{\mathcal{K} \in \mathfrak{K}}
\eta_{\mathcal{K}}
\left\lbrack
\mathds{1}_0\left(\sum_{k \in \mathcal{K}} y_k\right)
\binom{N}{\{y_j\}_{j \notin \mathcal{K}}}
\prod_{j \notin \mathcal{K}} \left(
\theta_j^{\mathcal{K}}
\right)^{y_j}
\right\rbrack  \\
\label{eq:zanim_pmf__all_zeros__mixture}
&\phantom{=}~+\eta_0\,\prod_{j=1}^{d}\,\mathds{1}_0(y_j)
\quad \mathbf{y} \in \bm{\Omega}_{d,N}^0,
\end{align}where
$\bm{\theta} = (\theta_1, \ldots, \theta_d)$, with $\theta_j \geq 0$ and $\sum_{j=1}^d\theta_j=1$,
$\theta_j^{\mathcal{K}} = \theta_j/(1 - \sum_{\ell \in \mathcal{K}}\,\theta_{\ell})$, and
$\bm{\zeta} = (\zeta_1, \ldots, \zeta_d)$, with
$\zeta_j \in \lbrack 0, 1\rbrack$.
The mixture weights $\bm{\eta}$ are functions of $\bm{\zeta}$ (see \autoref{def:etas}).
\end{theorem}
\begin{proof}[\textbf{\upshape Proof:}]
See \ref{app:zanim_pmf} for details.
\end{proof}

\begin{remark}
    In \autoref{theo:zanim_pmf}, $\bm{\Omega}_{d,N}^0 = \bm{\Omega}_{d,N} \cup \mathbf{0}_d$ is an expansion of the multinomial support, which accounts for the case where
    $y_j=0$ for all $j\in\{1,\ldots,d\}$. Although simultaneously observing zero counts for all categories may be a rarity in practice, the addition of the associated component \eqref{eq:zanim_pmf__all_zeros__mixture} is necessary to ensure the validity of the PMF.
\end{remark}

It is clear that the ZANIM distribution is a finite mixture of multinomials, along with two degenerate distributions for the cases where all counts are zero, corresponding to \eqref{eq:zanim_pmf__all_zeros__mixture}, and where all but one category are zero, corresponding to \eqref{eq:zanim_pmf__only_one_no_zero__mixture}. Furthermore, \eqref{eq:zanim_pmf__no_zero__mixture} represents a multinomial distribution with $d$ categories, $N$ trials, and probabilities $\bm{\theta} = (\theta_1, \ldots, \theta_d)$, while \eqref{eq:zanim_pmf__sets_of_zi_non_zi__mixture} represents multinomials of reduced dimension with $N$ trials and probabilities $\theta_j^{\mathcal{K}}$ for all categories $j \notin \mathcal{K}$. These components reflect the fact that counts in the remaining categories will inflate when zeros co-occur in fewer than $d-1$ categories, while the $N$-inflation components \eqref{eq:zanim_pmf__only_one_no_zero__mixture} capture the extreme cases in which exactly $d-1$ categories are zero-inflated.

We note, however, that a given $\mathbf{y}_i=\{y_{i1},\ldots,y_{id}\}$ can belong to as few as one and at most $2^{d-1}$ components, given the presence of indicator functions in the above PMF. Obviously, if there are no zeros in the given $\mathbf{y}_i$, we need only evaluate the purely multinomial component \eqref{eq:zanim_pmf__no_zero__mixture}. In the special case where $\mathbf{y}_i$ consists entirely of zeros, we evaluate only the purely degenerate component \eqref{eq:zanim_pmf__all_zeros__mixture}.
Otherwise, when $\mathbf{y}_i$ contains both zero and non-zero counts, we require the evaluation of $2^q$ components, where $q\in\lbrack 1,d-1\rbrack$ denotes the number of observed zeros, including the purely multinomial component, the reduced multinomial components, and (when $q=d-1$ exactly) the corresponding $N$-inflated component. This simplifies likelihood calculations by obviating the need to evaluate all $2^d$ components and highlights how relatively few components need to be evaluated when the number of zeros is low.

\begin{proposition}\label{prop:zanim_finite_mixture}
If $\mathbf{Y} \sim \operatorname{ZANIM}_d\lbrack N, \bm{\theta}, \bm{\zeta} \rbrack$, then $\mathbf{Y}$ has
the stochastic representation:
\begin{equation}\label{eq:zanim_finite_mixture}
\mathbf{Y} \sim \eta_0
\delta_{\mathbf{0}_d}(\cdot) + \sum_{j=1}^d\eta^{(j)}_N\left(N\mathbf{e}_d^{(j)}\right) +
\eta_d\operatorname{Multinomial}_d\lbrack N, \bm{\theta} \rbrack +
\sum_{\mathcal{K} \in \mathfrak{K}}\eta_{\mathcal{K}}\operatorname{Multinomial}_d\lbrack N, \bm{\theta}_0^{\mathcal{K}}\rbrack,
\end{equation}
where $\smash{\mathbf{e}_d^{(j)}=(\delta_{0}(y_1), \ldots, \delta_{0}(y_{j-1}), \delta_{1}(y_j), \delta_{0}(y_{j+1}),\ldots,\delta_{0}(y_{d}))}$
denotes the canonical basis vector of length $d$, with Dirac mass at $1$ in the $j$-th entry and Dirac masses at zero elsewhere, and
$\bm{\theta}_0^{\mathcal{K}}$ reflects the fact that the entries of $\{\theta_1^\mathcal{K},\ldots,\theta_d^\mathcal{K}\}$
are zero for all $j \in \mathcal{K}$.
\begin{proof}[\textbf{\upshape Proof:}]
Follows by identifying each mixture component in the ZANIM PMF.
\end{proof}
\end{proposition}
Our ZANIM naming convention comes from the observation that \eqref{eq:zanim_finite_mixture} contains degenerate components, which capture the case where $\mathbf{y}=\mathbf{0}_d$ and the cases where $y_j=N$. We also note that the mixture components in the ZANIM PMF are generated from a set of independent Bernoulli random variables, which leads to an alternative stochastic representation.
\begin{proposition}\label{prop:stochastic_representation_zanim}
If $\mathbf{Y} \sim \operatorname{ZANIM}_d\lbrack N, \bm{\theta}, \bm{\zeta} \rbrack$, then $\mathbf{Y}$ has
the stochastic representation:
\begin{align*}
(z_j \mid \zeta_j) &\sim \operatorname{Bernoulli}\lbrack 1 - \zeta_j \rbrack, \quad j\in\{1,\ldots, d\},\\
(\mathbf{Y} \mid N, \bm{\theta}, \mathbf{z}) &\sim
\begin{cases} \delta_{\mathbf{0}_d}(\cdot), & \textrm{if} \: z_{j} = 0 \: \forall j,\\
\operatorname{Multinomial}_d\left\lbrack N, z_1\dfrac{\theta_1}{1 - s}, \ldots, z_d\dfrac{\theta_d}{1 - s}\right\rbrack,
& \textrm{otherwise},
\end{cases}
\end{align*}where $\smash{s=\sum_{k=1}^d (1 - z_k)\theta_k}$.
\begin{proof}[\textbf{\upshape Proof:}]
Follows from the fact that the mixture weights from \autoref{def:etas} are products of Bernoulli probabilities.
\end{proof}
\end{proposition}

From the two stochastic representations above, we can easily generate values from the ZANIM distribution.
Obviously, the representation in \autoref{prop:stochastic_representation_zanim}
is more efficient.
We note that this representation has similarities to the zero-inflated logistic normal multinomial model of \citet{Zeng2023}. However, in their model, they do not consider the case where $z_{j}=0 \: \forall \: j$,
which has been studied in a specific space-time application by \citet{DouwesSchultz2024}.

\subsection{\texorpdfstring{Zero-and-$N$-inflated Dirichlet-multinomial distribution}{Zero-and-N-inflated Dirichlet-multinomial distribution}}\label{sec:zanidm}

Leveraging the hierarchical representation of the Dirichlet-multinomial distribution (henceforth DM)
through the compounding of the multinomial and Dirichlet distributions, \citet{Koslovsky2023} introduced
the zero-inflated Dirichlet-multinomial (ZIDM) distribution and provided a latent stochastic representation thereof (see \autoref{def:zanidm_stochastic_representation}).
Through our derivation of a novel probabilistic representation of ZIDM as a finite mixture
distribution, we suggest that the zero-and-$N$-inflated Dirichlet-multinomial (ZANIDM) distribution is a more appropriate name.

\begin{definition}\label{def:zanidm_stochastic_representation}
A random vector $\mathbf{Y}$ is said to follow a ZANIDM distribution if
it has the following stochastic representation:
\begin{align*}
(z_{j} \mid \zeta_{j}) & \sim \operatorname{Bernoulli}\lbrack 1 - \zeta_{j}\rbrack, \quad j \in \{1,\ldots, d\},\\
(\lambda_{j} \mid z_{j}, \alpha_{j}) &\sim (1 - z_{j})\delta_0(\cdot) + z_{j} \operatorname{Gamma}\lbrack \alpha_{j}, 1\rbrack, \\
(\mathbf{Y} \mid N, \bm{\theta}, \mathbf{z}) &\sim
\begin{cases} \delta_{\mathbf{0}_d}(\cdot), & \textrm{if} \: z_{j} = 0 \: \forall j,\\
\operatorname{Multinomial}_d\left\lbrack N, \theta_1, \ldots, \theta_d\right\rbrack,
& \textrm{otherwise},
\end{cases}\end{align*}
where $\theta_{j} = \lambda_{j} / \sum_{k=1}^{d}\lambda_{k}$.
In brief, we write $\mathbf{Y} \sim \operatorname{ZANIDM}_d\lbrack N, \bm{\alpha}, \bm{\zeta} \rbrack$,
where the parameters are: $N$, the number of trials;
$\bm{\zeta}=(\zeta_1, \ldots, \zeta_d)$, s.t. $\zeta_j\in \lbrack 0 , 1\rbrack$ is the
excess-of-zero parameter of category $j$; and
$\bm{\alpha} = (\alpha_1, \ldots, \alpha_d)$, the concentration parameters, with $\alpha_j > 0$.
\end{definition}

\begin{theorem}\label{theo:zanidm_pmf}
If $\mathbf{Y} \sim \operatorname{ZANIDM}_d\lbrack N, \bm{\alpha}, \bm{\zeta} \rbrack$, then $\mathbf{Y}$
is a finite mixture distribution with $2^d$ components and PMF given by
\begin{align}
\label{eq:zanidm_pmf__no_zero__mixture}
\Pr\lbrack \mathbf{Y} = \mathbf{y}; \bm{\alpha}, \bm{\zeta}\rbrack &= \eta_d
\dfrac{\Gamma\left(\alpha_s\right)\Gamma\left(N + 1\right)}{\Gamma\left(N + \alpha_s\right)}
\prod_{j=1}^d\left \lbrack \dfrac{\Gamma\left(y_j + \alpha_j\right)}{\Gamma\left(\alpha_j\right)\Gamma\left(y_j + 1\right)}
\right \rbrack\\
\label{eq:zanidm_pmf__only_one_no_zero__mixture}
&\phantom{=}~+
\sum_{j=1}^{d}
\eta_{N}^{(j)}
\left\lbrack
\mathds{1}_0\left(\sum_{k\colon k \neq j} y_{k} \right)
\right\rbrack \\
\label{eq:zanidm_pmf__sets_of_zi_non_zi__mixture}
&\phantom{=}~+
\sum_{\mathcal{K} \in \mathfrak{K}}
\eta_{\mathcal{K}}
\left\lbrack
\mathds{1}_0\left(\sum_{i \in \mathcal{K}} y_i\right)
\right\rbrack
\dfrac{\Gamma\alpha^{\mathcal{K}}_s)\Gamma(N + 1)}{\Gamma(N + \alpha^{\mathcal{K}}_s)}
\prod_{j \notin \mathcal{K}}\left\lbrack
\dfrac{\Gamma(y_j + \alpha_j)}{\Gamma(\alpha_j)\Gamma(y_j + 1)}
\right \rbrack\\
\label{eq:zanidm_pmf__all_zeros__mixture}
&\phantom{=}~+\eta_0\prod_{j=1}^{d}\mathds{1}_0(y_j),
\quad \mathbf{y} \in \bm{\Omega}_{d,N}^0,
\end{align}
where $\alpha_s = \sum_{j=1}^d\alpha_j$,
$\alpha^{\mathcal{K}}_s = \sum_{j \notin \mathcal{K}}\alpha_j$, and the mixture weights
$\bm{\eta}$ are as given in \autoref{def:etas}.
\begin{proof}[\textbf{\upshape Proof:}]
Marginalising out the latent variables $z_{j}$ and $\lambda_j$ in \autoref{def:zanidm_stochastic_representation}
and accounting for the fact that $y_j = 0$ when $\lambda_j=0(z_{j}=1)$ yields the desired result.
See \ref{app:zanidm_pmf} for details.
\end{proof}
\end{theorem}

Recasting ZIDM as ZANIDM under our finite mixture framework enables a fuller characterisation of the distribution which highlights, in particular, that the PMF also incorporates degenerate components.
Notably, the $N$-inflation components in \eqref{eq:zanidm_pmf__only_one_no_zero__mixture}, as well as the case where $\mathbf{y}=\mathbf{0}_d$ in \eqref{eq:zanidm_pmf__all_zeros__mixture}, are identical to their ZANIM counterparts in
\eqref{eq:zanim_pmf__only_one_no_zero__mixture} and \eqref{eq:zanim_pmf__all_zeros__mixture}, respectively.
However, the remaining component distributions differ from the ZANIM distribution in \eqref{eq:zanidm_pmf__no_zero__mixture} is a DM distribution and \eqref{eq:zanidm_pmf__sets_of_zi_non_zi__mixture} represents DM distributions of reduced dimension, in contrast to the multinomial distribution in \eqref{eq:zanim_pmf__no_zero__mixture} and the sets of reduced multinomial distributions in \eqref{eq:zanim_pmf__sets_of_zi_non_zi__mixture} under the ZANIM distribution.
Similarly, the finite mixture stochastic representation for ZANIDM is obtained by appropriately
replacing the multinomial distributions in \autoref{prop:zanim_finite_mixture} with DM distributions.
\begin{proposition}\label{prop:zanidm_finite_mixture}
If $\mathbf{Y} \sim \operatorname{ZANIDM}_d\lbrack N, \bm{\alpha}, \bm{\zeta} \rbrack$, then $\mathbf{Y}$ has
the stochastic representation:
\begin{equation}\label{eq:zanidm_finite_mixture}
\mathbf{Y} \sim \eta_0\delta_{\mathbf{0}_d}(\cdot) +
\sum_{j=1}^d\eta^{(j)}_N\left(N\mathbf{e}_d^{(j)}\right) + \eta_d\operatorname{DM}\lbrack N, \bm{\alpha} \rbrack +
\sum_{\mathcal{K} \in \mathfrak{K}}\eta_{\mathcal{K}}\operatorname{DM}\lbrack N, \bm{\alpha}_0^{\mathcal{K}}\rbrack,
\end{equation}where $\bm{\alpha}_0^{\mathcal{K}}$ denotes a concentration parameter vector of reduced dimension which reflects the fact that entries of $\bm{\alpha}$ are irrelevant for all $j \in \mathcal{K}$.
\begin{proof}[\textbf{\upshape Proof:}]
Follows by identifying each mixture component in the ZANIDM PMF.
\end{proof}
\end{proposition}

\begin{remark}
    Although ZANIM and ZANIDM are finite mixture distributions with $2^d$ components, there are only $2d$
    parameters in each case, since the component weights are functions of $\bm{\zeta}$ and their respective $\bm{\theta}$ and $\bm{\alpha}$
    parameters fully determine the non-degenerate components.
\end{remark}

\begin{remark}\label{remark:zanim_zanidm_relationship}
The DM distribution arises from compounding the multinomial distribution with a probability vector which follows a Dirichlet distribution.
The concentration of its random success probabilities around their mean is governed by $\alpha_s=\sum_{j=1}^d\alpha_j$.
This variability diminishes as $\alpha_s \rightarrow \infty$ while keeping the
proportions $\alpha_j/\alpha_s$ constant. Consequently, the Dirichlet distribution collapses to a degenerate distribution and the DM tends toward a multinomial distribution with fixed probabilities.
Hence, we can relate the ZANIM and ZANIDM distributions by noting that, subject to keeping the corresponding parameters constant,
the DM component in the ZANIDM distribution approaches a multinomial distribution with parameters $\alpha_j/\alpha_s$ as $\alpha_s\rightarrow \infty$, while the reduced DM components tend toward their multinomial counterparts with parameters $\alpha_j/\alpha_s^{\mathcal{K}}$ as $\alpha_s^{\mathcal{K}}\rightarrow\infty$. The remaining degenerate components are already common to both distributions.
\end{remark}

\begin{remark}\label{remark:special_cases}
When $d=2$, the two-dimensional vector $\mathbf{Y}$
can be represented by a single random variable.
In light of the above derivations, we note that two special cases arise from the ZANIM and ZANIDM distributions in such instances; namely,
the zero-and-endpoint-inflated binomial (ZANIB) \citep{Deng2015,Tian2015} and zero-and-endpoint-inflated beta-binomial (ZANIBB) distributions, respectively. To the best of our knowledge, the ZANIBB distribution was first proposed by \citet{Sweeney2016}. The ZANIB and ZANIBB PMFs can be easily deduced by setting $d=2$ in \autoref{theo:zanim_pmf} and \autoref{theo:zanidm_pmf}, respectively, and discarding the corresponding reduced multinomial and reduced Dirichlet-multinomial components in \eqref{eq:zanim_pmf__sets_of_zi_non_zi__mixture} and \eqref{eq:zanidm_pmf__sets_of_zi_non_zi__mixture}, which are exclusive to $d>2$ settings.
\end{remark}

 \section{Properties of ZANIM and ZANIDM}
\label{sec:properties}

We outline some basic probabilistic properties of the ZANIM and ZANIDM distributions.
Notably, the properties that follow are consequences of the fact that both the
ZANIM and ZANIDM distributions can be seen as finite mixture distributions.
As described in \autoref{theo:zanim_pmf} and \autoref{theo:zanidm_pmf}, both models have $K=2^d$ mixture components and the corresponding mixture weights $\bm{\eta}$ are functions of the excess-of-zero parameters $\bm{\zeta}$
(see \autoref{def:etas}).

\subsection{\texorpdfstring{Marginal distribution of $Y_j$}{Marginal distribution of Yj}}\label{sec:marginal_Y_j}
When $\mathbf{Y}$ follows either the ZANIM or ZANIDM distribution, whose respective PMFs are given in
\eqref{eq:zanim_pmf__no_zero__mixture}--\eqref{eq:zanim_pmf__all_zeros__mixture} and
\eqref{eq:zanidm_pmf__no_zero__mixture}--\eqref{eq:zanidm_pmf__all_zeros__mixture}, the marginal PMF of the $j$-th element is itself a mixture.
Under both distributions, obtaining the marginal PMF of $Y_j$ involves summing over the support of all other elements of the random vector $\mathbf{Y}$ while fixing $Y_j=k$.
We use the well-known results that the multinomial and DM distributions have binomial and beta-binomial marginals, respectively, and introduce the following set notation to help write explicit formulas for the marginal PMFs.
\begin{definition}\label{def:set_Sj}
Let $\mathfrak{S}_j = \{\mathcal{S}_j \subseteq \{1,\ldots j-1, j+1 \ldots,d\}; 1 \leq \lvert\mathcal{S}_j\rvert \leq d-2\}$ represent the set of all subsets $\mathcal{S}_j$ of $\{1,\ldots,d\}\setminus \{j\}$
with cardinality $\lvert\mathcal{S}_j\rvert$ between $1$ and $d - 2$ for a given category $j$, such that the subsets $\mathcal{S}_j$ are obtained by excluding the cases which contain the $j$-th index from the subsets $\mathcal{K}$ described in \autoref{def:set_K}.
Note that $\lvert \mathfrak{S}_j \rvert = 2^{d-1} - 2$.
\end{definition}

Performing the required summation over each mixture component in the ZANIM and ZANIDM PMFs is
reasonably straightforward.
Under ZANIM, the `non-zero' component in \eqref{eq:zanim_pmf__no_zero__mixture}
and the `reduced dimension' components in \eqref{eq:zanim_pmf__sets_of_zi_non_zi__mixture}
yield weighted binomial distributions, with weights given by $\eta_d$ and $\eta_{\mathcal{S}_j} = \prod_{k\in\mathcal{S}_j}\zeta_k\prod_{j\notin\mathcal{S}_j}(1-\zeta_j)$, respectively.
The corresponding ZANIDM components in \eqref{eq:zanidm_pmf__no_zero__mixture} and
\eqref{eq:zanidm_pmf__sets_of_zi_non_zi__mixture}, yield similarly weighted beta-binomial distributions.
Recall that the remaining components are common to both distributions.
Firstly, regarding the $N$-inflated components in \eqref{eq:zanim_pmf__only_one_no_zero__mixture}
and \eqref{eq:zanidm_pmf__only_one_no_zero__mixture},
we note that marginalising the $N$-inflation component corresponding to the $j$-th category contributes a
degenerate mass at $N$ with probability $\eta_N^{(j)}$, while the remaining $N$-inflated components each
contribute a degenerate mass at $0$ with probability $\eta_N^{(k)}$, for $k \neq j$.
Secondly, the purely degenerate components in \eqref{eq:zanim_pmf__all_zeros__mixture} and
\eqref{eq:zanidm_pmf__all_zeros__mixture} also contribute a degenerate mass at $0$ with probability $\eta_0$. Therefore, the marginal distribution of $Y_j$ is degenerate at $0$ with probability $\eta_0 + \sum_{k\neq j} \eta_N^{(k)}$, which simplifies to $\zeta_j$. Finally, we obtain the marginal PMF of $Y_j$ under both distributions by combining the contributions of the associated marginal mixture components.

\begin{proposition}\label{prop:marginal_yj_zanim}
If $\mathbf{Y} \sim \operatorname{ZANIM}_d\lbrack N, \bm{\theta}, \bm{\zeta}\rbrack$, then the marginal distribution of $Y_j$ is
\begin{align}
\Pr\lbrack Y_j = k\rbrack =
\begin{cases}
\zeta_j + \eta_d\left(1 - \theta_j\right)^N +
\sum_{\mathcal{S}_j\in\mathfrak{S}_j}\eta_{\mathcal{S}_j}
\left(1 - \theta_{j}^{\mathcal{S}_j}\right)^N, &
\textrm{if} \quad k=0,\\
\eta_N^{(j)} + \eta_d\theta_j^N +
\sum_{\mathcal{S}_j\in\mathfrak{S}_j}\eta_{\mathcal{S}_j}\left(\theta_{j}^{\mathcal{S}_j}\right)^N, &
\textrm{if} \quad k=N,\\
\eta_d p_B\left(k; N, \theta_j\right) +
\sum_{\mathcal{S}_j\in\mathfrak{S}_j}\eta_{\mathcal{S}_j}p_B\left(k; N, \theta_{j}^{\mathcal{S}_j}\right),
&\textrm{if} \quad k\in\{1,\ldots, N-1\},\nonumber
\end{cases}
\end{align}
where $p_B(k; N, \theta)$ denotes the binomial PMF
and $\theta_{j}^{\mathcal{S}_j} = \theta_j/(1 - \sum_{\ell \in \mathcal{S}_j} \theta_\ell)$.
\begin{proof}[\textbf{\upshape Proof:}]
Follows from \autoref{def:set_Sj} and appropriate summation over the ZANIM PMF.
\end{proof}
\end{proposition}

\begin{proposition}\label{prop:marginal_yj_zanidm}
If $\mathbf{Y} \sim \operatorname{ZANIDM}_d\lbrack N, \bm{\alpha}, \bm{\zeta}\rbrack$, then the marginal distribution of $Y_j$ is
\begin{align}
\Pr\lbrack Y_j = k\rbrack =
\begin{cases}
\zeta_j + \eta_d\dfrac{\operatorname{B}\left(\alpha_j, N + \alpha^{(j)}_s\right)}{\operatorname{B}\left(\alpha_j, \alpha^{(j)}_s\right)} +
\sum_{\mathcal{S}_j\in\mathfrak{S}_j}\eta_{\mathcal{S}_j}
\dfrac{\operatorname{B}\left(\alpha_j, N + \alpha_s^{\mathcal{S}_j}\right)}{\operatorname{B}\left(\alpha_j, \alpha_s^{\mathcal{S}_j}\right)},&
\textrm{if} \quad k=0,\\
\eta_N^{(j)} +
\eta_d\dfrac{\operatorname{B}\left(\alpha_j, N + \alpha^{(j)}_s\right)}{\operatorname{B}\left(\alpha_j, \alpha^{(j)}_s\right)}+
\sum_{\mathcal{S}_j\in\mathfrak{S}_j}
\eta_{\mathcal{S}_j}\dfrac{\operatorname{B}\left(\alpha_j + N, \alpha_s^{\mathcal{S}_j}\right)}{\operatorname{B}\left(\alpha_j, \alpha_s^{\mathcal{S}_j}\right)},&
\textrm{if} \quad k=N,\\
\eta_d
p_{BB}\left(k; N, \alpha_{j}, \alpha^{(j)}_s\right) +
\sum_{\mathcal{S}_j\in \mathfrak{S}_j}\eta_{\mathcal{S}_j}p_{BB}\left(k; N, \alpha_{j}, \alpha_s^{\mathcal{S}_j}\right),&
\textrm{if} \quad k\in\{1,\ldots, N-1\},\nonumber
\end{cases}
\end{align}
where $p_{BB}(k; N, \alpha_j, \alpha^{(j)}_s)$ denotes the beta-binomial PMF with
$\alpha^{(j)}_s = \sum_{\ell \neq j} \alpha_\ell$ and $\operatorname{B}(a, b) = \Gamma(a)\Gamma(b)/\Gamma(a+b)$ is the Beta function.
Similarly, $\alpha_s^{\mathcal{S}_j} = \sum_{\ell \notin \mathcal{S}_j} \alpha_\ell$, where $\ell \notin \mathcal{S}_j$ implies the indices in $\{1, \ldots, d\}\setminus\{j\}$ that are not in $\mathcal{S}_j$.
\begin{proof}[\textbf{\upshape Proof:}]
Follows from \autoref{def:set_Sj} and appropriate summation over the ZANIDM PMF.
\end{proof}
\end{proposition}

\begin{remark}\label{remark:prob0_probN}
For the marginal probability distributions presented in \autoref{prop:marginal_yj_zanim} and
\autoref{prop:marginal_yj_zanidm}, both $\Pr\lbrack Y_j = 0\rbrack \rightarrow \zeta_j$ and $\Pr\lbrack Y_j = N\rbrack \rightarrow \eta_N^{(j)}$ as $N \rightarrow \infty$.
\end{remark}

From \autoref{prop:marginal_yj_zanim} and \autoref{prop:marginal_yj_zanidm}, we identify that
the marginal distribution of $Y_j$ under ZANIM (or ZANIDM) is a finite mixture containing $2^{d-1} + 1$ components.
These mixture components are either degenerate at zero, degenerate at $N$, or follow binomial (or beta-binomial)
distributions.
\autoref{fig:comparison_zanim_zanidm_pmf} shows some of the types of behaviour the marginal PMFs of ZANIM and ZANIDM can have in a three-dimensional setting, where the marginals of the random vector $\mathbf{Y}$ each have category-specific parameters. Each marginal distribution is specified such that the expectation is identical under both distributions (see Section \ref{sec:moments} for details on their moments).
\begin{figure}[H]
    \centering
    \includegraphics[width=1\linewidth]{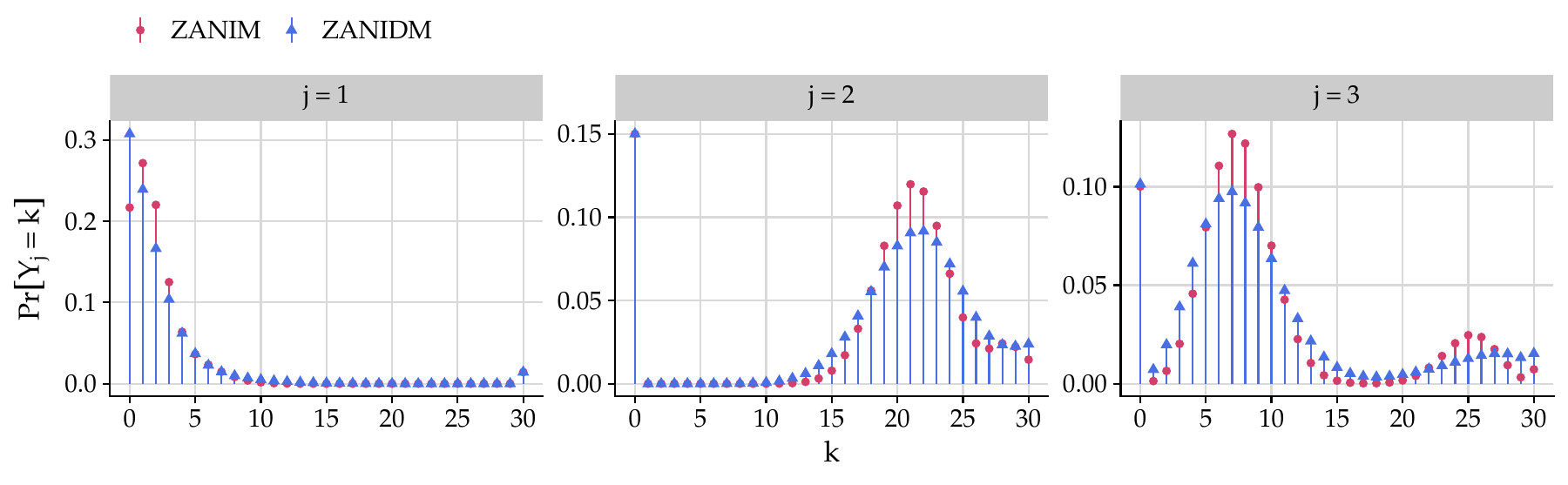}
    \caption{Marginal PMFs of ZANIM (red circles) and ZANIDM (blue triangles) with respective parameters $\bm{\theta} \in \{0.05, 0.70, 0.25\}$ for ZANIM and $\bm{\alpha} \in \{2.0, 28.0, 10.0\}$ for ZANIDM, along with $\bm{\zeta} \in \{0.05, 0.15, 0.10\}$ and $N=30$ trials in each case.}
    \label{fig:comparison_zanim_zanidm_pmf}
\end{figure}The first marginal, $Y_1$, has a large spike at $k=0$ under both distributions, although this consists not only of structural zeros, but also many sampling zeros. This is understandable, given that the $\theta_1$ and $\alpha_1$ parameters take their lowest values among the three categories when $j=1$, and the corresponding zero-inflation parameter is also low $(\zeta_1=0.05)$.
Although most of the probability mass is concentrated at lower values of $k$, there is also slight but nonetheless visible $N$-inflation at $k=N=30$, since the corresponding mixture weight, $\eta_N^{(1)} = (1 - \zeta_1)\zeta_2\zeta_3 = 0.00675$, is non-zero. Interestingly, the ZANIDM marginal is more right-skewed and overdispersed than that of ZANIM, which also results in a larger spike at zero.
The second marginal, $Y_2$, has a noticeable spike at $k=0$ also, given the higher $\zeta_2=0.15$, but peaks at higher $k$ values given
that the $\theta_2$ and $\alpha_2$ parameters take their highest values among the three categories when $j=2$.
As there is little mass assigned to values of $1 \le k < 15$, this reflects a scenario where counts of zero reflect failure to detect a phenomenon which is more common.
Finally, the nature of ZANIM as a finite mixture is most apparent for the third marginal, $Y_3$, given that both zero-inflation and moderate $N$-inflation are evident in addition to two other modes. Under ZANIM, the larger mode at $k=6$ is attributable to the purely multinomial component of the mixture, while the smaller mode at $k=24$ relates to the set of
multinomial distributions of reduced dimension which capture `sets of zero-inflation' (in this case pairs) in the other categories.
Regarding ZANIDM, we can clearly see a heavy right tail, which can be explained by the finite mixture components and the corresponding overdispersed nature of the beta-binomial distributions.

\subsection{Moments}\label{sec:moments}

We briefly review a generic result that will be used for the derivations.
Let $\bm{\tau}=\{\tau_1,\ldots,\tau_K\}$ be a set of mixture weights, such that $\tau_k$ reflects a generic indexed mixture component
without regard to whether that component relates to any particular case of zero-inflation, $N$-inflation, or otherwise.
Let $W$ denote a discrete random variable indicating the mixture component, i.e.,
taking values in $\{1,\ldots, K\}$ with corresponding probabilities $\bm{\tau}$.
Then, the expected value of $g(\mathbf{Y})$ can be expressed as
$\mathbb{E}\lbrack g(\mathbf{Y})\rbrack = \sum_{k=1}^{K}\tau_k\mathbb{E}\lbrack g(\mathbf{Y}) \mid W = k \rbrack$. Thus, since both distributions can be expressed as multivariate finite mixture distributions
with $K = 2^d$ components, we can easily obtain their moments. We recall that the component distributions are either degenerate random vectors (at $\delta_{\mathbf{0}_d}$ or $N\mathbf{e}_d^{(j)}$), or multinomial/DM
random vectors (including some of reduced dimension); see the stochastic representations in
\eqref{eq:zanim_finite_mixture} and \eqref{eq:zanidm_finite_mixture}.

\begin{definition}
By analogy with the sets $\mathfrak{K}$ and $\mathfrak{S}_j$ described in \autoref{def:set_K} and \autoref{def:set_Sj}, respectively, let $\mathfrak{R}_{jh} = \{\mathcal{R}_{jh} \subseteq \{1,\ldots, j-1, j+1, \ldots, h-1, h+1, \ldots, d\}; 1 \leq \lvert\mathcal{R}_{jh}\rvert \leq d-2\}$ represent the set of all subsets $\mathcal{R}_{jh}$ of $\{1,\ldots,d\}\setminus \{j,h\}$ for a given pair of categories $j \ne h$.
\end{definition}

\begin{proposition}\label{prop:zanim_moments}
Let $\mathbf{Y} \sim \operatorname{ZANIM}_d\lbrack N, \bm{\theta}, \bm{\zeta} \rbrack$,
then the mean and variance of the $j$-th entry of the random vector $\mathbf{Y}$, i.e.,
the random variable $Y_j$, are given by
\begin{align*}
\mathbb{E}\lbrack Y_j\rbrack &= N\left(\eta^{(j)}_N + \eta_d\theta_j +
\sum_{\mathcal{S}_j \in \mathfrak{S}_j}\eta_{\mathcal{S}_j}\theta_{j}^{\mathcal{S}_j}\right),
\\
\operatorname{Var}\lbrack Y_j \rbrack &= \eta_N^{(j)}N^2 +
\eta_dN \theta_j \left( 1 + \theta_j\left(N - 1\right)\right)
+
\sum_{\mathcal{S}_j \in \mathfrak{S}_j}\eta_{\mathcal{S}_j}N\theta_{j}^{\mathcal{S}_j}
\left(1 + \theta_{j}^{\mathcal{S}_j} \left(N - 1\right) \right)- N^2\left(\eta^{(j)}_N + \eta_d\theta_j +
\sum_{\mathcal{S}_j \in \mathfrak{S}_j}\eta_{\mathcal{S}_j}\theta_{j}^{\mathcal{S}_j}\right)^2.
\end{align*}
The covariance between the random variables $(Y_j, Y_h), j \ne h$, of the random vector $\mathbf{Y}$ is
\begin{align*}
\operatorname{Cov}\lbrack Y_j, Y_h\rbrack
&=
N(N - 1)\left(\eta_d\theta_j\theta_h + \sum_{\mathcal{R}_{jh}\in\mathfrak{R}_{jh}}\eta_{\mathcal{R}_{jh}}\theta_{j}^{\mathcal{R}_{jh}}\theta_{h}^{\mathcal{R}_{jh}}\right) -N^2\left(\eta^{(j)}_N + \eta_d\theta_j +
\sum_{\mathcal{S}_j \in \mathfrak{S}_j}\eta_{\mathcal{S}_j}\theta_{j}^{\mathcal{S}_j}\right)
\left(\eta^{(h)}_N + \eta_d\theta_h +
\sum_{\mathcal{S}_h \in \mathfrak{S}_h}\eta_{\mathcal{S}_h}\theta_{h}^{\mathcal{S}_h}\right),
\end{align*}
where $\theta_{j}^{\mathcal{R}_{jh}} = \theta_j / (1 - \sum_{\ell \in \mathcal{R}_{jh}} \theta_\ell)$.
\begin{proof}[\textbf{\upshape Proof:}]
Follows from the moment properties of finite mixture distributions.
\end{proof}
\end{proposition}

\begin{proposition}
Let $\mathbf{Y} \sim \operatorname{ZANIDM}_d\lbrack N, \bm{\alpha}, \bm{\zeta} \rbrack$,
then the mean and variance of the $j$-th entry of the random vector $\mathbf{Y}$, i.e.,
the random variable $Y_j$, are given by
\begin{align*}
\mathbb{E}\lbrack Y_j \rbrack &= N\left(\eta^{(j)}_N + \eta_d\dfrac{\alpha_j}{\alpha_s} +
\sum_{\mathcal{S}_j \in \mathfrak{S}_j}\eta_{\mathcal{S}_j}
\dfrac{\alpha_j}{\alpha_j + \alpha_s^{\mathcal{S}_j}} \right),\\
\mathrm{Var}\lbrack Y_j \rbrack &=
\eta_N^{(j)}N^2 +
\eta_d\left(
\dfrac{N\alpha_j\left(N \left(1 + \alpha_j\right) + \alpha_s^{(j)}\right)}{\left(\alpha_j + \alpha_s^{(j)}\right)\left(1 + \alpha_j + \alpha_s^{(j)}\right)}\right)+
\sum_{\mathcal{S}_j \in \mathfrak{S}_j}\eta_{\mathcal{S}_j}
\left(
\dfrac{N\alpha_j\left(N \left(1 + \alpha_j\right) + \alpha_s^{{\mathcal{S}_j}}\right)}{
\left(\alpha_j + \alpha_s^{{\mathcal{S}_j}}\right)\left(1 + \alpha_j + \alpha_s^{{\mathcal{S}_j}}\right)}\right) \\
&\phantom{=}~-
N^2\left(\eta^{(j)}_N + \eta_d\dfrac{\alpha_j}{\alpha_s} +
\sum_{\mathcal{S}_j \in \mathfrak{S}_j}\eta_{\mathcal{S}_j}
\dfrac{\alpha_j}{\alpha_j + \alpha_s^{\mathcal{S}_j}} \right)^2.
\end{align*}
The covariance between the random variables $(Y_j, Y_h), j \ne h$, of the random vector $\mathbf{Y}$ is
\begin{align*}
\operatorname{Cov}\lbrack Y_j, Y_h\rbrack &=
N\eta_d\dfrac{\alpha_j\alpha_h}{\alpha_s}\left(N - \dfrac{N + \alpha_s}{1 + \alpha_s}\right)+N\sum_{\mathcal{R}_{jh} \in \mathfrak{R}_{jh}}\eta_{\mathcal{R}_{jh}}
\dfrac{\alpha_j\alpha_h}{\left(\alpha_j + \alpha_h + \alpha_s^{\mathcal{R}_{jh}}\right)^2}
\left(N -
\dfrac{N + \alpha_j+\alpha_h+\alpha_s^{\mathcal{R}_{jh}}}{1 + \alpha_j+\alpha_h+\alpha_s^{\mathcal{R}_{jh}}}
\right) \\
&\phantom{=}~-N^2
\left(\eta^{(j)}_N + \eta_d\dfrac{\alpha_j}{\alpha_s} +
\sum_{\mathcal{S}_j \in \mathfrak{S}_j}\eta_{\mathcal{S}_j}
\dfrac{\alpha_j}{\alpha_j + \alpha_s^{\mathcal{S}_j}} \right) \left(\eta^{(h)}_N + \eta_d\dfrac{\alpha_h}{\alpha_s} +
\sum_{\mathcal{S}_h \in \mathfrak{S}_h}\eta_{\mathcal{S}_h}
\dfrac{\alpha_h}{\alpha_h + \alpha_s^{\mathcal{S}_h}} \right),
\end{align*}
where $\alpha_s^{\mathcal{R}_{jh}} = \sum_{\ell \notin \mathcal{R}_{jh}} \alpha_\ell$.
\begin{proof}[\textbf{\upshape Proof:}]
Follows from the moment properties of finite mixture distributions.
\end{proof}
\end{proposition}
\autoref{tab:theoretical_moments} compares the theoretical moments of the ZANIM
and ZANIDM distributions. We also report the dispersion index $\operatorname{DI}\lbrack Y_j \rbrack = \operatorname{Var}\lbrack Y_j \rbrack/\mathbb{E}\lbrack Y_j \rbrack$ and the zero-inflation index
$\operatorname{ZI}\lbrack Y_j \rbrack = 1 + \log\left(\Pr\lbrack Y_j = 0\rbrack\right) / \mathbb{E}\lbrack Y_j \rbrack$,
for each category $j$ under both distributions. See \citet{Puig2006} for details of these indices.
In this setting, both distributions yield identical means by construction.
The variances under ZANIDM are higher than those under ZANIM, which highlights ZANIDM's greater
flexibility in modelling overdispersion.
However, the $\operatorname{DI}\lbrack Y_j \rbrack$ indices imply that both distributions can handle overdispersion, which can arise due to
zero-inflation. Although ZANIDM's values for this index are greater, it is notable that the ZANIM distribution can still capture some degree of overdispersion.
The $\operatorname{ZI}\lbrack Y_j \rbrack$ index, reflecting the degree of zero-inflation, is slightly
higher under ZANIDM when $j=1$, but otherwise the values match for both distributions. However, we stress that excess zeros may not always be structural in nature; they can also arise due to overdispersion, as the two phenomena are linked.
\begin{table}[H]
\centering
\captionsetup{width=0.6\textwidth}
\caption{Comparison of the theoretical moments of ZANIM and ZANIDM, with
$\bm{\theta} \in \{0.05, 0.70, 0.25\}$ for ZANIM and $\bm{\alpha} \in \{2.0, 28.0, 10.0\}$ for ZANIDM, along with $\bm{\zeta} \in \{0.05, 0.15, 0.10\}$ and $N=30$ trials in each case.}
\vskip-0.3cm \noindent\rule{.6\textwidth}{0.4pt}\smallskip
\label{tab:theoretical_moments}
\begin{tabular*}{.6\textwidth}{@{\hspace{0.01\textwidth}\extracolsep{\fill}}lrrrrr@{\hspace{0.01\textwidth}}}
& Distribution & $\mathbb{E}\lbrack Y_j\rbrack$ & $\mathrm{Var}\lbrack Y_j \rbrack$ & $\mathrm{DI}\lbrack Y_j \rbrack$ & $\mathrm{ZI}\lbrack Y_j \rbrack$ \\
\midrule
\multirow{2}{*}{$j = 1$} & ZANIM & 2.320 & 14.326 & 6.174 & 0.341 \\
   & ZANIDM & 2.320 & 16.392 & 7.064 & 0.492 \\ \hline
\multirow{2}{*}{$j = 2$} & ZANIM & 18.496 & 69.178 & 3.740 & 0.897 \\
   & ZANIDM & 18.496 & 72.723 & 3.932 & 0.897 \\ \hline
\multirow{2}{*}{$j = 3$} & ZANIM & 9.161 & 50.409 & 5.502 & 0.749 \\
   & ZANIDM & 9.161 & 54.658 & 5.966 & 0.750
\end{tabular*}
\rule{.6\textwidth}{0.4pt}
\end{table}\autoref{tab:theoretical_covariance} gives the theoretical covariances between different categories
under the ZANIM and ZANIDM distributions, with the same parameter settings.
Notably, both distributions are capable of accommodating both negative and positive dependence, unlike the standard multinomial and DM distributions under which the covariances between two elements of the random vector $\mathbf{Y}$ are strictly non-positive by definition. The usual covariances of the multinomial and DM distributions can be recovered from the expressions derived above when $\bm{\zeta}=\mathbf{0}_d$.
\begin{table}[H]
\centering
\captionsetup{width=0.6\textwidth}
\caption{Comparison of the theoretical covariances of ZANIM and ZANIDM, with
$\bm{\theta} \in \{0.05, 0.70, 0.25\}$ for ZANIM and $\bm{\alpha} \in \{2.0, 28.0, 10.0\}$ for ZANIDM, along with $\bm{\zeta} \in \{0.05, 0.15, 0.10\}$ and $N=30$ trials in each case.}
\label{tab:theoretical_covariance}
\vskip-0.3cm\rule{.6\textwidth}{0.4pt}\smallskip

\begin{tabular*}{.6\textwidth}{@{\hspace{0.01\textwidth}\extracolsep{\fill}}lrr@{\hspace{0.01\textwidth}}}
$\operatorname{Cov}\lbrack Y_j, Y_h\rbrack$ & ZANIM & ZANIDM \\
\midrule
$\operatorname{Cov}\lbrack Y_1, Y_2\rbrack$ & $-16.416$ & $-17.097$ \\
  $\operatorname{Cov}\lbrack Y_1, Y_3\rbrack$ & 2.143 & 0.758 \\
  $\operatorname{Cov}\lbrack Y_2, Y_3\rbrack$ & $-52.346$ & $-55.210$ \\
\end{tabular*}
\rule{.6\textwidth}{0.4pt}
\end{table}The theoretical features of the ZANIM and ZANIDM distributions highlight their flexibility for modelling
count-compositional data with an excess of zeros, while also accommodating overdispersion and covariance
structures that can capture both positive and negative dependence.
The fact that both distributions can be represented as finite mixtures allows them to be flexible in this regard, even if their constituent component distributions, by themselves, are not. In the Supplementary Material, we show how moment generating functions can also be derived for both distributions, again using the properties of finite mixtures. \section{Bayesian inference for ZANIM and ZANIDM}\label{sec:inference}

We develop Bayesian inference frameworks for estimating the parameters of the ZANIM and ZANIDM distributions.
Their respective inference schemes are based on the likelihood functions defined in
\eqref{eq:zanim_pmf__no_zero__mixture}--\eqref{eq:zanim_pmf__all_zeros__mixture} and
\eqref{eq:zanidm_pmf__no_zero__mixture}--\eqref{eq:zanidm_pmf__all_zeros__mixture}.
These functions involve complex mixture likelihoods where the mixing proportion
$\boldsymbol{\eta}$ depend on
the zero-inflation parameters $\zeta_1,\ldots,\zeta_d$. As the dimension $d$ increases, computing the likelihood
becomes computationally intensive.
To address this, we exploit
the stochastic representations of the distributions and consider
data augmentation strategies, thereby simplifying the posterior distributions and enabling efficient sampling. In each case, we assume access to an i.i.d.~random sample of size $n$ denoted by
$\mathbf{y} = (\mathbf{y}_1, \ldots, \mathbf{y}_n)$, where $\mathbf{y}_i = (y_{i1}, \ldots, y_{id})$.
We allow the number of trials, a fixed and known parameter given by $N_i=\sum_{j=1}^d y_{ij}$, to be observation-specific, such that $\mathbf{Y}_i \sim \operatorname{ZANIM}_d\lbrack N_i, \bm{\theta}, \bm{\zeta}\rbrack$ or $\mathbf{Y}_i \sim \operatorname{ZANIDM}_d\lbrack N_i, \bm{\alpha}, \bm{\zeta}\rbrack$.

\subsection{ZANIM}

Inference for the ZANIM parameters $\bm{\theta}$ and $\bm{\zeta}$ is based on
the stochastic representation given in \autoref{prop:stochastic_representation_zanim}. In the Supplementary Material, we show that augmenting the ZANIM distribution
with the latent variables $z_{1},\ldots,z_{d}$ and $(\phi \mid \mathbf{y}, \mathbf{z}) \sim \operatorname{Gamma}\lbrack N, \sum_{j=1}^d\lambda_j z_{j}\rbrack$ enables recovery of the zero-inflated
augmented likelihood in \eqref{eq:pmf_zanim_augmented} from which the
ZANIM distribution was initially derived, and give the full MCMC algorithm.
Based on a random sample $\mathbf{y}$, the further augmented ZANIM likelihood using the latent variables
$\bm{\phi} = (\phi_1, \ldots, \phi_n)$ and $\mathbf{z} = (\mathbf{z}_1, \ldots, \mathbf{z}_n)$, with
$\mathbf{z}_i = (z_{i1}, \ldots, z_{id})$, is given by
\begin{align*}
\mathcal{L}(\bm{\lambda}, \bm{\zeta}; \mathbf{y}, \bm{\phi}, \mathbf{z})
&\propto
\prod_{i=1}^n\prod_{j=1}^d\left\{\left\lbrack(1 - z_{ij})\zeta_j + z_{ij}(1-\zeta_j)e^{-\phi_i\lambda_j}\right\rbrack^{\mathds{1}_0(y_{ij})} \times\left\lbrack z_{ij}(1-\zeta_j) \lambda_j^{y_{ij}}e^{-\phi_i\lambda_j}\right\rbrack^{1-\mathds{1}_0(y_{ij})}\right\}\\
&=\prod_{j=1}^d\left\{\zeta_j^{n - t_j}(1 - \zeta_j)^{t_j} \times \lambda_j^{r_{j}} e^{-s_j\lambda_j}\right\},
\end{align*}
where
$t_j = \sum_{i=1}^{n}z_{ij}, r_{j} = \sum_{i=1}^ny_{ij}z_{ij}$, and
$s_j = \sum_{i=1}^n\phi_iz_{ij}$ play the role of conditional sufficient statistics for category $j$.
We can now see that the augmented likelihood factors into a product of beta and gamma terms,
and that the category-specific parameters are independent. This implies that inference procedures
can be performed independently for each category. Thus, we can consider a joint prior for
$(\zeta_j, \lambda_j)$ as a product of two independent priors which exhibit conjugacy properties, i.e,
$\zeta_j \sim \operatorname{Beta}\lbrack a_j, b_j\rbrack$ and
$\lambda_j \sim \operatorname{Gamma}\lbrack c_j, d_j\rbrack,$
for $j \in \{1,\ldots, d\}$.
Thus, given the augmented data $(\mathbf{y}, \bm{\phi}, \mathbf{z})$, the full conditional distribution of $\zeta_j$ is given by
\begin{equation}\label{eq:posterior_zeta}
(\zeta_j \mid \mathbf{y}, \bm{\phi}, \mathbf{z}) \sim \operatorname{Beta}\lbrack n - t_j + a_j,  t_j + b_j\rbrack,
\end{equation}
while the full conditional distribution of $\lambda_j$ is given by
\begin{equation}\label{eq:posterior_lambda}
(\lambda_j \mid \mathbf{y}, \bm{\phi}, \mathbf{z}) \sim \operatorname{Gamma}\lbrack r_j + c_j, s_j + d_j\rbrack.
\end{equation}Finally, we note that the full conditional distribution of $z_{ij}$ is given by
\begin{equation}\label{eq:bernoulli_augmented}
(z_{ij} \mid y_{ij}, \phi_i) \sim
\begin{cases}
\operatorname{Bernoulli}\left\lbrack \dfrac{(1-\zeta_j)e^{-\phi_i\lambda_j}}{\zeta_j + (1-\zeta_j)e^{-\phi_i\lambda_j}}\right\rbrack, & \textrm{if} \: y_{ij} = 0,\\
1, & \textrm{otherwise}.
\end{cases}\end{equation}

\subsection{ZANIDM}\label{sec:zanidm_inference}

As per \citet{Koslovsky2023}, we exploit the stochastic representation of the ZANIDM distribution given in \autoref{def:zanidm_stochastic_representation}.
We further introduce the latent variables
\[(\phi_i \mid \mathbf{y}_i, \bm{\lambda}_i) \sim \operatorname{Gamma}\left \lbrack N_i, \sum_{j=1}^d\lambda_{ij}\right \rbrack,\:i \in \{1,\ldots, n\},\quad N_i = \sum_{j=1}^d y_{ij}.\]
Given the latent variables $\lambda_{ij}$, $z_{ij}$, and $\phi_i$, along with the observed
vector $\mathbf{y}_i$, the augmented likelihood of the $i$-th observation factors into $d$ independent terms, as per ZANIM, as follows
\begin{align*}
p(\mathbf{y}_i, \bm{\lambda}_i, \mathbf{z}_i, \phi_{i} \mid \bm{\alpha}, \bm{\zeta})
&=
p(\mathbf{y}_i \mid \bm{\lambda}_i)
p(\phi_i \mid \mathbf{y}_i, \bm{\lambda}_i)
\prod_{j=1}^d\left\lbrack p(\lambda_{ij} \mid z_{ij})p(z_{ij})
\right\rbrack \\
&=\dfrac{\phi_i^{N_i -1}}{\Gamma(N_i)}\binom{N}{y_1, \ldots, y_d}
\prod_{j=1}^d\left\lbrack
(1 - \zeta_j)^{z_{ij}}\zeta_{j}^{1 - z_{ij}}
(1 - z_{ij})^{\mathds{1}_0(\lambda_{ij})}z_{ij}^{1 - \mathds{1}_0(\lambda_{ij})}
\right\rbrack\prod_{j=1}^d \left\lbrack
\left(\dfrac{\lambda_{ij}^{y_{ij}+\alpha_j - 1}e^{-\lambda_{ij}(1 + \phi_i) }}{\Gamma(\alpha_j)}\right)^{1 - \mathds{1}_0(\lambda_{ij})}
\right\rbrack.
\end{align*}Thus, the inference over the parameters $\bm{\zeta}$ and $\bm{\alpha}$ can be performed independently.

\citet{Koslovsky2023} proposed a method to perform Bayesian inference when the parameters $\bm{\alpha}$ and
$\bm{\zeta}$ can depend on covariates (which we do not consider here), which relies on so-called
`expand and contract' moves (effectively transdimensional Metropolis-Hastings (MH) steps) to jointly update the
latent variables $z_{ij}$ and $\lambda_{ij}$. We propose an alternative collapsed Gibbs sampling approach that improves the efficiency by enabling fast conjugate updates for both quantities. By avoiding joint updates, we obviate the need to change the dimension of the parameter space as the MCMC algorithm proceeds.
Specifically, we note that it is
easy to obtain the distribution of
$z_{ij}$ unconditional on $\lambda_{ij}$, then update $\lambda_{ij}$ conditional on $z_{ij}$. We sketch our proposals below, but provide more details in  the Supplementary Material, including the full MCMC algorithm.

We first discuss the updates of $\lambda_{ij}$ and $z_{ij}$.
Note that the joint distribution of $\lambda_{ij}$ and $z_{ij}$ given the observed data $y_{ij}$ and
the latent variable $\phi_i$ is
\[
p(\lambda_{ij}, z_{ij} \mid y_{ij}, \phi_i) \propto (1 - \zeta_j)^{z_{ij}}\zeta_{j}^{1 - z_{ij}}
\left\lbrack(1 - z_{ij})\delta_0(\lambda_{ij}) + z_{ij}
\dfrac{\lambda_{ij}^{y_{ij} + \alpha_j - 1}e^{-\lambda_{ij}(1 + \phi_i) }}{\Gamma(\alpha_j)}\right\rbrack.
\]
An easy away to avoid the complicated expand and contract approach of \citet{Koslovsky2023} when updating
$\lambda_{ij}$ and $z_{ij}$ is to take advantage of the marginalisation of the joint distribution
$p(\lambda_{ij}, z_{ij} \mid y_{ij}, \phi_i)$ over $\lambda_{ij}$.
In doing so, we obtain
\[
p(z_{ij} \mid y_{ij}, \phi_i) \propto (1 - z_{ij})\zeta_{j} + z_{ij}(1 - \zeta_j)
\dfrac{(1 +\phi_i)^{-(y_{ij} + \alpha_j)}\Gamma(y_{ij} + \alpha_j)}{\Gamma(\alpha_j)}.
\]
For a given $j$, we know that $z_{ij}=1$ with probability $1$ when $y_{ij} > 0$. Conversely, when $y_{ij} = 0$, we have that
\[p(z_{ij} \mid y_{ij}=0, \phi_i) \propto \left\lbrack(1 - \zeta_j)(1 +\phi_i)^{-\alpha_j}\right\rbrack^{z_{ij}}
\zeta_{j}^{1 - z_{ij}}.\]
Therefore, the collapsed conditional distribution of $z_{ij}$ is
\begin{equation}\label{eq:zanidm_z_update}
(z_{ij} \mid y_{ij}, \phi_i) \sim
\begin{cases}
\operatorname{Bernoulli}\left\lbrack
\dfrac{(1 - \zeta_j) (1 +\phi_i)^{-\alpha_j}}{\zeta_j + (1 - \zeta_j) (1 +\phi_i)^{-\alpha_j}}
\right\rbrack, & \textrm{if} \: y_{ij} = 0,\\
1, & \textrm{otherwise},
\end{cases}
\end{equation}
and the full conditional distribution of $\lambda_{ij}$ is
\begin{equation}\label{eq:zanidm_lambda_update}
(\lambda_{ij} \mid y_{ij}, z_{ij}, \phi_i) \sim \begin{cases}
\operatorname{Gamma}\lbrack \alpha_j + y_{ij}, 1 + \phi_i \rbrack, & \textrm{if} \quad z_{ij}=1,\\
0, & \textrm{if} \quad  z_{ij}=0,
\end{cases}
\end{equation}
which is recognisable as a zero-augmented gamma distribution. Thus, straightforward Gibbs updates are available for both $\lambda_{ij}$ and $z_{ij}$, without requiring the joint expand and contract updates performed by \citet{Koslovsky2023}.

As regards the parameter $\zeta_j$, its full conditional distribution is given by
\[
p(\zeta_j \mid \mathbf{y},  \mathbf{z}, \bm{\phi})
\propto (1 - \zeta_j)^{t_{j}}\,\zeta_{j}^{n - t_{j}} \times p(\zeta_j),
\]
where $t_j = \sum_{i=1}^n z_{ij}$, as per ZANIM.
However, we stress that the updates for the latent $z_{ij}$ under ZANIM and ZANIDM, in \eqref{eq:bernoulli_augmented} and \eqref{eq:zanidm_z_update}, respectively, are distinct.
We clearly have the kernel of a Bernoulli distribution, as before; by assuming its conjugate
prior, $\zeta_j \sim \operatorname{Beta}\lbrack a_j, b_j \rbrack$, we have that
the full conditional distribution for $\zeta_j$ is
\begin{equation}\label{eq:posterior_zeta_zanidm}
(\zeta_j \mid \mathbf{y}, \mathbf{z}, \bm{\phi}) \sim \operatorname{Beta}\lbrack n - t_j + a_j, t_j + b_j \rbrack.
\end{equation}

The full conditional distribution for $\alpha_j$ can be written as
\begin{equation}\label{eq:alpha_j_full_conditional}
p(\alpha_j \mid \mathbf{y}, \bm{\lambda}, \mathbf{z}, \bm{\phi}) \propto
\dfrac{1}{\Gamma(\alpha_j)^{t_j}}\,\exp\left\lbrack \alpha_j\,\sum_{i\colon \lambda_{ij}>0} \ln \lambda_{ij} \right\rbrack \times p(\alpha_j),
\end{equation}
However, this does not resemble the kernel of a known distribution.
To overcome this, we consider and evaluate the performance of three approaches.
The first assumes a gamma prior for
$\alpha_j \sim \operatorname{Gamma}\lbrack c_j, d_j\rbrack$ in conjunction with \eqref{eq:alpha_j_full_conditional}, which results in the full conditional target
\begin{equation}\label{eq:target_alpha_da_ptn}
\pi(\alpha_j)  \coloneqq p(\alpha_j \mid \mathbf{y}, \bm{\lambda}, \mathbf{z}, \bm{\phi}) \times p(\alpha_j) \propto \dfrac{\alpha_j^{c_j - 1}}{\Gamma(\alpha_j)^{t_j}}
\exp\left\lbrack -\alpha_j \left(d_j - \sum_{i\colon \lambda_{ij}>0} \ln \lambda_{ij} \right)\right\rbrack.
\end{equation}
We then apply a data augmentation scheme proposed by \citet{Hamura2023}, which performs a MH step
with independent power-truncated-normal (PTN) proposals.

The second and third approaches consider the re-parameterisation $\ln \alpha_j = \beta_j$
and assume a Normal prior, $\beta_j \sim \operatorname{Normal}\lbrack m_j, s_j^2 \rbrack$,
resulting in the full conditional target
\begin{equation}\label{eq:target_beta_mh_ss}
\pi(\beta_j)
\propto \dfrac{1}{\Gamma(e^{\beta_j})^{t_j}} \exp\left\lbrack e^{\beta_j} \sum_{i\colon \lambda_{ij}>0} \ln \lambda_{ij} + \beta_js_j^{-2} (m_j - 0.5\beta_j)\right\rbrack.
\end{equation}
To sample from $\pi(\beta_j)$, we consider two well-known general schemes; a MH algorithm where the proposals
follow a Gaussian random walk, as used by \citet{Koslovsky2023}, and slice sampling with the stepping-out and shrinkage procedures as proposed by
\citet{Neal2003}. Note that our use of random walk MH still differs from \citet{Koslovsky2023} by virtue of our novel updates for $\lambda_{ij}$ and $z_{ij}$. Further details on all proposed sampling schemes for $\bm{\alpha}$ are provided in the Supplementary Material. \section{Simulation studies}\label{sec:simulation_study}

Our simulation experiments first compare the MCMC schemes described in Section \ref{sec:zanidm_inference} for inferring ZANIDM's parameters and secondly illustrate the practical utility of both the ZANIM and ZANIDM distributions when dealing with zero-inflation in count-compositional data.

\subsection{Comparison of MCMC algorithms for ZANIDM}\label{sec:simstudy1}

In this simulation exercise, we compare the MCMC schemes discussed in Section \ref{sec:zanidm_inference}.
Our goals are: (i) to demonstrate that
our collapsed Gibbs sampling approach for updating
$z_{ij}$ and $\lambda_{ij}$ offers superior efficiency and inferential performance compared to the
joint updates proposed by \citet{Koslovsky2023}; and
(ii) to evaluate different approaches for sampling the $\bm{\alpha}$ parameters.
To this end, we consider four approaches: the algorithm by \citet{Koslovsky2023},
which is available via the \textsf{R} package \texttt{ZIDM} on the author's GitHub
repository at \url{https://github.com/mkoslovsky/ZIDM},
and our three proposed methods discussed in Section \ref{sec:zanidm_inference}.
In brief, these variations differ in how they sample $\bm{\alpha}$ as follows:
DA-PTN utilises data-augmentation and MH with PTN proposals
introduced by \citet{Hamura2023},
MH-RW employs random-walk MH for $\ln \alpha_j$,
SS implements slice sampling using stepping-out and shrinkage procedures for $\ln \alpha_j$.
For all but the DA-PTN approach, we consider the prior $\ln \alpha_j \sim \operatorname{Normal}\lbrack 0, 5\rbrack$. For DA-PTN, we match the hyper-parameters of the $\operatorname{Gamma}\lbrack c_j, d_j\rbrack$ prior such that $\mathbb{E}\lbrack \ln \alpha_j\rbrack \approx 0$ and
$\operatorname{Var}\lbrack \ln \alpha_j\rbrack \approx 5$.
Regarding $\zeta_j$, we use the $\operatorname{Beta}\lbrack 1, 1\rbrack$ prior for our DA-PTN, MH-RW, and SS implementations.
As regards \texttt{ZIDM}, we recall that this implementation samples from
$\ln\,((1 - \zeta_j)/{\zeta_j})$ with a $\operatorname{Normal}\lbrack 0 , 5\rbrack$ prior, by default.

We consider a scenario with $d = 20$ categories. We simulate $R = 50$ replicates from the ZANIDM distribution, varying the
sample sizes $(n)$ and numbers of trials $(N)$ as $ n, N \in \{50, 200, 500\}$.
Our setup closely mirrors the one considered by \citet{Koslovsky2023}, where the zero-inflation parameters,
$\bm{\zeta}=(\zeta_1, \ldots, \zeta_{20})$, were randomly drawn from $\operatorname{Uniform}\lbrack 0.0, 0.5\rbrack$ and $\bm{\alpha}=(\alpha_1, \ldots, \alpha_{20})$ were
randomly drawn via $\ln \alpha_j \sim \operatorname{Uniform}\lbrack -2.3, 2.3\rbrack$. Here, the true values of the $\bm{\zeta}$ and $\bm{\alpha}$ parameters range from $0.006$ to $0.490$ and $0.149$ to $8.031$, respectively. For all MCMC algorithms, we use $51{,}000$ iterations, discard the first $1{,}000$ draws, and thin every $50$-th draw to reduce the dependency between them. We thereby obtain $1{,}000$ valid posterior samples.

To measure efficiency, we compute the average effective sample size (ESS) ratio --- which is
obtained by dividing the ESS by the number of valid posterior samples ---
for both the $\bm{\alpha}$ and $\bm{\zeta}$ parameters across the $R=50$ replicates. The results are displayed using box-plots
in \autoref{fig:frame_ess_bias_cp_zanidm_N_200} (panel A).
We quantify parameter recovery for $\bm{\alpha}$ and $\bm{\zeta}$ using overall relative bias based on the posterior mean, and the overall coverage probability
of the $95\%$ credible interval. Letting $\bm{\vartheta}=(\vartheta_1, \ldots, \vartheta_{d})$ denote the
true values of the parameter vector of interest, we compute these metrics as follows:
\[
\operatorname{Bias}(\bm{\vartheta}) = \dfrac{1}{Rd}\sum_{r=1}^{R} \sum_{j=1}^{d}
\left(\dfrac{\mathbb{E}\lbrack \vartheta_j \mid \mathbf{y}^{(r)}\rbrack}{\vartheta_j}-1\right),\quad
\operatorname{CP}_{95\%}(\bm{\vartheta}) = \dfrac{1}{Rd}\sum_{r=1}^{R}\sum_{j=1}^{d}
\mathds{1}\left(\vartheta_j \in \operatorname{CI}_{95\%}\lbrack \vartheta_j \mid \mathbf{y}^{(r)}\rbrack\right),
\]
where $\mathbb{E}\lbrack\vartheta_j \mid \mathbf{y}^{(r)}\rbrack$ and
$\operatorname{CI}_{95\%}\lbrack \vartheta_j \mid \mathbf{y}^{(r)}\rbrack$
are the posterior mean and $95\%$ credible interval, respectively, for the $\vartheta_j$ parameter on the $r$-th replicate.

\autoref{fig:frame_ess_bias_cp_zanidm_N_200} shows the results for $N = 200$ across the various sample sizes.
Panel A displays box-plots of the ESS ratios for $\bm{\alpha}$ and $\bm{\zeta}$. It is evident that the approaches
which incorporate Gibbs updates for $z_{ij}$ and $\lambda_{ij}$ have consistently higher ESS ratios than \texttt{ZIDM}, with the exception of
the MH-RW method when the sample size is $50$. This provides strong evidence that our proposal enhances efficiency.
Panels B and C present the bias, $\operatorname{Bias}(\bm{\vartheta})$, and the $95\%$ coverage probability,
$\operatorname{CP}_{95\%}(\bm{\vartheta})$, respectively. The results show that \texttt{ZIDM}
exhibits the highest bias and lowest coverage for both parameters.
While the DA-PTN approach achieves the lowest bias, it is accompanied by a notably large coverage probability.
In contrast, the MH-RW and SS methods display intermediate bias levels and coverage probabilities close to the nominal value.
\begin{figure}[H]
    \centering
    \includegraphics[width=1\linewidth]{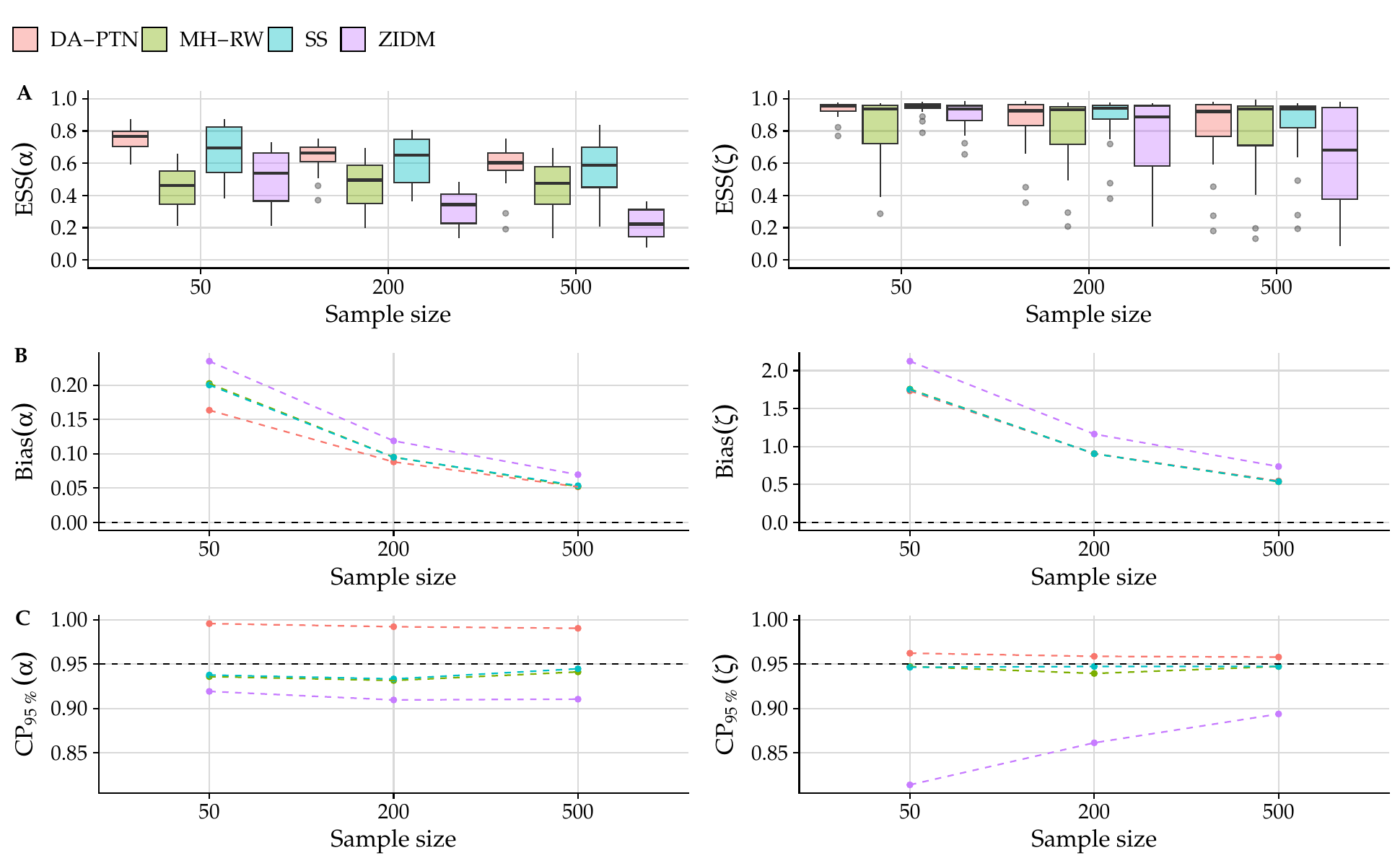}
    \caption{
    Comparison of efficiency and parameter recovery of the $\bm{\alpha}$ (left) and $\bm{\zeta}$ (right) parameters for different ZANIDM inference schemes. All metrics are averaged over the $d=20$ categories and $R=50$ replicates simulated from ZANIDM with $N = 200$ trials and
    varying sample size $\{50, 200, 500\}$.
    \textbf{A}: effective sample size ratio; \textbf{B}: overall relative bias based on the posterior mean;
    \textbf{C}: overall coverage probability of the $95\%$ credible interval.}
    \label{fig:frame_ess_bias_cp_zanidm_N_200}
\end{figure}
The results with lower ($N=50$) and higher ($N=500$) numbers of trials are omitted for brevity.
The conclusions about the performance of each inference scheme across all three metrics are
broadly in line with those drawn from \autoref{fig:frame_ess_bias_cp_zanidm_N_200}.
As $N$ varies, only the magnitude of the bias changes; the other metrics are stable and the relative rankings of each approach are unchanged.

\subsection{Simulated data analysis examples}\label{sec:simstudy2}

This simulation exercise shows the utility of both distributions
for addressing zero-inflation in count-compositional data. We simulate two data sets, each
containing $500$ observations. The data-generating processes (DGPs) are based on the ZANIM
and ZANIDM distributions, following their respective stochastic representations given in
\autoref{prop:stochastic_representation_zanim} and \autoref{def:zanidm_stochastic_representation}.
As the parameter values used here match the imbalanced configurations used for \autoref{fig:comparison_zanim_zanidm_pmf}, \autoref{tab:theoretical_moments}, and \autoref{tab:theoretical_covariance}, they are particularly challenging.
Specifically, we have $\bm{\theta} = (0.05, 0.70, 0.25)$ under ZANIM and $\bm{\alpha} = (2.0, 28.0, 10.0)$ under ZANIDM, with $\bm{\zeta} = (0.05, 0.15, 0.10)$ and $N = 30$ trials in each case.

When fitting ZANIM, we run our MCMC scheme
for $11{,}000$ iterations, with the first $1{,}000$ discarded as burn-in and a thinning interval of $10$ applied, and specify the following priors: $\lambda_j \sim \operatorname{Gamma}\lbrack 0.1, 0.1\rbrack$ and
$\zeta_j \sim \operatorname{Beta}\lbrack 1, 1\rbrack$.
For ZANIDM, we run our MCMC scheme
for $110{,}000$ iterations, with the first $10{,}000$ discarded as burn-in and a thinning interval of $100$ applied.
This setup helps to ensure reliable posterior inference and
reduce autocorrelation in the chains, particularly those for $\bm{\alpha}$.
The prior for $\zeta_j$ is set as per the ZANIM model and we use the DA-PTN approach to infer
$\alpha_j$, with its $\operatorname{Gamma}\lbrack c_j,d_j\rbrack$ prior elicited as per Section \ref{sec:simstudy1}. \autoref{tab:posterior_summaries_sim_2} presents the posterior summaries
for the model parameters.
\begin{table}[H]
\centering
\captionsetup{width=\textwidth}
\caption{
Posterior means, lower (LCI) and upper (UCI) limits of $95\%$ credible intervals, and effective sample size (ESS) ratios for the parameters of the ZANIM and ZANIDM models.
We report the posterior summaries for each model under two data-generating processes (DGPs), which are based on the ZANIM and ZANIDM distributions. For each DGP, $500$ samples are simulated from the corresponding distribution using imbalanced parameter configurations.
}
\label{tab:posterior_summaries_sim_2}
\vskip-0.3cm\hrule\smallskip
\begin{tabular*}{\textwidth}{@{\hspace{0.01\textwidth}\extracolsep{\fill}}llcrrrrr@{\hspace{0.01\textwidth}}}
DGP & Model & Parameter & Mean & $95\%$ LCI & $95\%$ UCI &  ESS ratio \\
\midrule
\multirow{12}{*}{\shortstack[l]{ZANIM:\\$\bm{\theta}\in \{0.05, 0.70, 0.25\},$\\$\bm{\zeta} \in \{0.05, 0.15, 0.10\}$,\\ $N=30$ trials.}} & \multirow{6}{*}{ZANIM} & $\theta_1$ & 0.047 & 0.044 & 0.051 & 0.808 \\
   &  & $\theta_2$ & 0.706 & 0.698 & 0.714 & 1.037 \\
   &  & $\theta_3$ & 0.246 & 0.239 & 0.253 & 1.039 \\
   &  & $\zeta_1$ & 0.025 & 0.001 & 0.070 & 0.489 \\
   &  & $\zeta_2$ & 0.140 & 0.111 & 0.172 & 0.930 \\
   &  & $\zeta_3$ & 0.122 & 0.096 & 0.151 & 0.833 \\
\cmidrule{2-8} & \multirow{6}{*}{ZANIDM} & $\alpha_1$ & 3.859 & 1.456 & 13.026 & 0.054 \\
   &  & $\alpha_2$ & 56.607 & 18.809 & 219.522 & 0.055 \\
   &  & $\alpha_3$ & 19.734 & 6.735 & 72.113 & 0.054 \\
   &  & $\zeta_1$ & 0.011 & 0.000 & 0.050 & 0.581 \\
   &  & $\zeta_2$ & 0.140 & 0.112 & 0.171 & 0.933 \\
   &  & $\zeta_3$ & 0.120 & 0.094 & 0.151 & 0.865 \\
\hline
\multirow{12}{*}{\shortstack[l]{ZANIDM:\\$\bm{\alpha}\in \{2.0, 28.0, 10.0\},$\\$\bm{\zeta} \in \{0.05, 0.15, 0.10\}$,\\$N=30$ trials.}} & \multirow{6}{*}{ZANIM} & $\theta_1$ & 0.053 & 0.049 & 0.058 & 0.989 \\
   &  & $\theta_2$ & 0.693 & 0.684 & 0.702 & 1.140 \\
   &  & $\theta_3$ & 0.254 & 0.247 & 0.262 & 1.064 \\
   &  & $\zeta_1$ & 0.205 & 0.148 & 0.258 & 0.915 \\
   &  & $\zeta_2$ & 0.127 & 0.100 & 0.156 & 0.812 \\
   &  & $\zeta_3$ & 0.096 & 0.073 & 0.122 & 1.032 \\
\cmidrule{2-8}
& \multirow{6}{*}{ZANIDM} & $\alpha_1$ & 1.241 & 0.787 & 2.301 & 0.148 \\
   &  & $\alpha_2$ & 18.822 & 11.284 & 35.420 & 0.145 \\
   &  & $\alpha_3$ & 6.829 & 4.130 & 12.985 & 0.133 \\
   &  & $\zeta_1$ & 0.025 & 0.001 & 0.095 & 0.476 \\
   &  & $\zeta_2$ & 0.129 & 0.101 & 0.160 & 0.902 \\
   &  & $\zeta_3$ & 0.093 & 0.068 & 0.120 & 0.827 \\
\end{tabular*}
\hrule
\end{table}Notably, both models closely recover the true values of the $\bm{\zeta}$ parameters when the data are generated from the ZANIM distribution. However, while the inference for $\bm{\theta}$ under the ZANIM model is satisfactory, the inference for $\bm{\alpha}$ under the ZANIDM model is poor, as indicated by wide credible intervals and low ESS ratios. Conversely, when the ZANIDM distribution is used as the DGP, the inference for $\bm{\zeta}$ under the ZANIM model is poor. The overestimation of $\zeta_1$ is especially notable and can be attributed to the overdispersed nature of the ZANIDM distribution.
The inference for $\bm{\alpha}$ under the ZANIDM model is improved in this case; the true values now fall within the $95\%$ credible intervals and the ESS ratios are larger. We note that the alternative MH-RW and slice sampling schemes yield similar results for the ZANIDM model, though \texttt{ZIDM} differs more substantially. For brevity, we defer these results, and those for additional simulations with balanced parameter settings, to the Supplementary Material.

To further compare the models, we compute the expected log-predictive density (ELPD)
which requires evaluating the log-PMFs of ZANIM and ZANIDM derived in Section \ref{sec:model_derivation}.
We stress that such likelihood-based model-selection criteria would not be feasible without first deriving
these finite mixture PMFs.
We use the Pareto smoothed importance sampling (PSIS) introduced by \citet{Vehtari2017} and available via the \textsf{R} package \texttt{loo}.
Given the posterior draws of the model parameters, denoted by $\{\bm{\vartheta}^{(m)}\}_{m=1}^{M}$, where $m\in\{1,\ldots,M\}$ indexes the number of valid posterior samples, the estimate of ELPD based on PSIS is defined by
\[
\widehat{\operatorname{elpd}} = \sum_{i=1}^n \ln \left(\dfrac{\sum_{m=1}^M w_i^{(m)} p(y_i \mid \bm{\vartheta}^{(m)}) }{\sum_{m=1}^M w_i^{(m)}} \right),
\]
where $w_i^{(m)}$ are the PSIS weights and $p(y_i \mid \bm{\vartheta}^{(m)})$ is
the model likelihood evaluated at the observation $y_i$.
The higher the ELPD, the better the model.

\autoref{tab:elpd_sim2} gives the ELPD results for different models under both DGPs. We
also include the multinomial and DM distributions for comparison purposes, for which we use Stan
via the \textsf{R} package \texttt{cmdstanr} \citep{cmdstanr} in each case.
As expected, the ELPD favours the distribution used to generate the data,
although the ZANIM and ZANIDM models obtain a similar ELPD when the data are generated from the ZANIM distribution.
Interestingly, when the data are generated from the ZANIDM distribution, we observe that the ZANIM model outperforms the DM model,
suggesting that accounting for zero-inflation improves the fit more than the overdispersion which
distinguishes the DM and multinomial distributions. However, we recall that the $\mathrm{DI}\lbrack Y_j\rbrack$ indices are similar for both DGPs (see \autoref{tab:theoretical_moments}).

\begin{table}[H]
\centering
\captionsetup{width=0.6\textwidth}
\caption{
Bayesian model evaluation metrics for different models with data simulated under two data-generating processes (DGPs) based on the ZANIM and ZANIDM distributions.
We report the expected log-predictive density $(\widehat{\operatorname{elpd}})$ and its standard error ($\operatorname{se}(\widehat{\operatorname{elpd}})$).
For each DGP, $500$ samples are simulated from the corresponding distribution using imbalanced parameter configurations.}
\label{tab:elpd_sim2}
\vskip-0.3cm\rule{.6\textwidth}{0.4pt}\smallskip

\begin{tabular*}{.6\textwidth}{@{\hspace{0.01\textwidth}\extracolsep{\fill}}llrr@{\hspace{0.01\textwidth}}}
DGP & Model & $\widehat{\operatorname{elpd}}$ & $\operatorname{se}(\widehat{\operatorname{elpd}})$
\\
  \midrule
\multirow{4}{*}{ZANIM} & ZANIM & $-2055.124$ & 20.778 \\
   & ZANIDM & $-2124.029$ & 16.156 \\
   & DM & $-2754.835$ & 26.139 \\
   & Multinomial & $-4751.385$ & 264.647 \\
  \hline
  \multirow{4}{*}{ZANIDM} & ZANIDM & $-2225.018$ & 20.995 \\
   & ZANIM & $-2306.174$ & 35.110 \\
   & DM & $-2662.285$ & 27.945 \\
   & Multinomial & $-4764.473$ & 266.561 \\
\end{tabular*}
\rule{.6\textwidth}{0.4pt}
\end{table}
For both the ZANIM and ZANIDM models, each panel in \autoref{fig:ppc_sim}
illustrates the mean and the $95\%$ CI of the posterior predictive distribution
(represented by red and blue error bars) compared with the empirical distribution of the observed count
$Y_j$ (depicted by grey bars) for each category $j\in \{1,2,3\}$.
To enhance the visualisation, we report the relative frequency and compare the empirical and posterior estimates thereof.
In Panel A, the data are generated from the ZANIM distribution and, as expected, the fitted ZANIM model closely
aligns with the observed data. In contrast, the ZANIDM model fails to capture certain patterns, particularly
for the component $j=2$. Conversely, when the ZANIDM distribution is used as the data-generating process, as shown
in Panel B, the ZANIDM model provides a better fit, effectively capturing the behaviour of the observed data, while the ZANIM model markedly deviates from the observed data, particularly when $\alpha_j$ is large.
\begin{figure}[H]
    \centering
    \includegraphics[width=\linewidth]{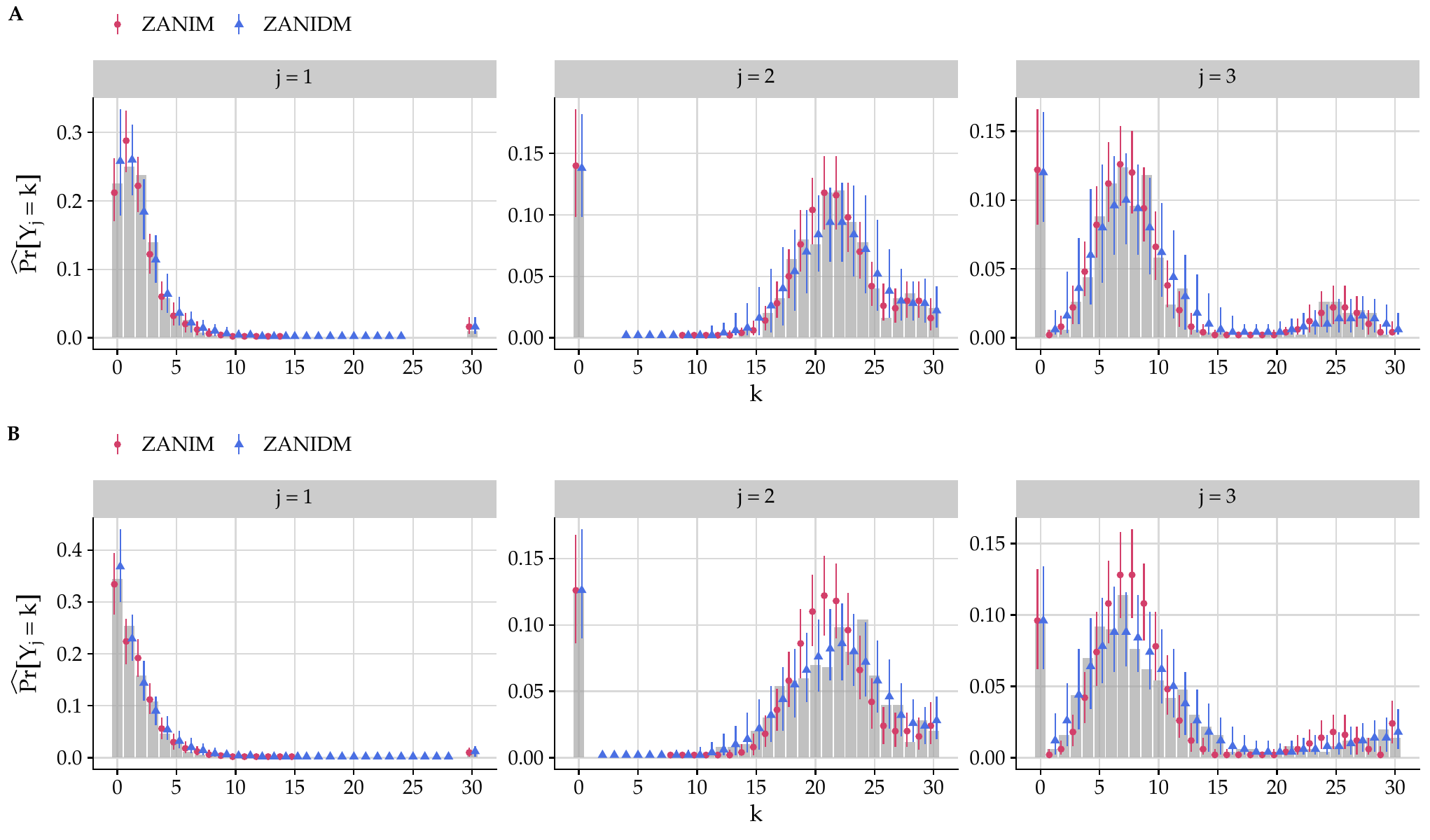}
    \caption{Empirical relative frequency estimates (grey bars) of the observed categories $y_j$, with the corresponding model estimates (where available) from the posterior predictive distributions of ZANIM (red circles) and ZANIDM (blue triangles). The points represent the means and the error-bars represent the corresponding $95\%$ prediction intervals. \textbf{A}: DGP from ZANIM. \textbf{B}: DGP from ZANIDM.}
    \label{fig:ppc_sim}
\end{figure} \section{Analyses of microbiome data}\label{sec:microbiome}

As a further demonstration of the ZANIM and ZANIDM distributions, we use them both to model a publicly-available
human gut microbiome dataset first studied by \citet{Wu2011}.
The data contain counts for $d=28$ genera-level operational taxonomic units obtained from 16S rRNA sequencing on $n=98$ individuals.
Across the entire dataset, $32.6\%$ of the observed counts are zero and the observation-specific numbers of trials $N_i$
range from $1{,}183$ to $15{,}447$. Notably, all taxa exhibit varying degrees of overdispersion, with the empirical
$\operatorname{DI}\lbrack Y_j\rbrack$ indices ranging from $2.92$ (\textit{Actinomycineae}) to $4418.20$ (\textit{Prevotella}).

We run our MCMC schemes for $50{,}000$ iterations, with the first $40{,}000$ discarded as burn-in and a thinning interval of $10$
applied. All prior distributions are specified as per Section \ref{sec:simstudy2}.
Our C\texttt{++} implementations of the ZANIM and ZANIDM inference schemes take approximately $1.5$ and $16.5$ seconds, respectively.
Furthermore, the ESS \textit{per second} is indicative of rapid mixing: when averaged across the taxa, we obtain $757.49$ and $766.10$ for the ZANIM parameters $\bm{\theta}$ and $\bm{\zeta}$,  respectively, and values of $34.65$ and $47.80$ for the respective ZANIDM parameters $\bm{\alpha}$ and $\bm{\zeta}$. Allowing for the difference in runtimes, the lower values for $\bm{\alpha}$ can be explained by the difficulties, outlined in Section \ref{sec:zanidm_inference}, in sampling ZANIDM's concentration parameters.

\autoref{fig:posterior_chosen_taxa} shows trace plots of all parameters of interest under both models for the aforementioned \textit{Actinomycinae} and \textit{Prevotella} taxa. These trace plots --- along with those for the remaining taxa, which we defer to the Supplementary Material --- are indicative of satisfactory convergence and affirm that our MCMC schemes remain stable, even in the presence of widely-varying numbers of trials and imbalanced levels of zero-inflation and overdispersion.

\begin{figure}[H]
    \centering
    \includegraphics[width=1\linewidth]{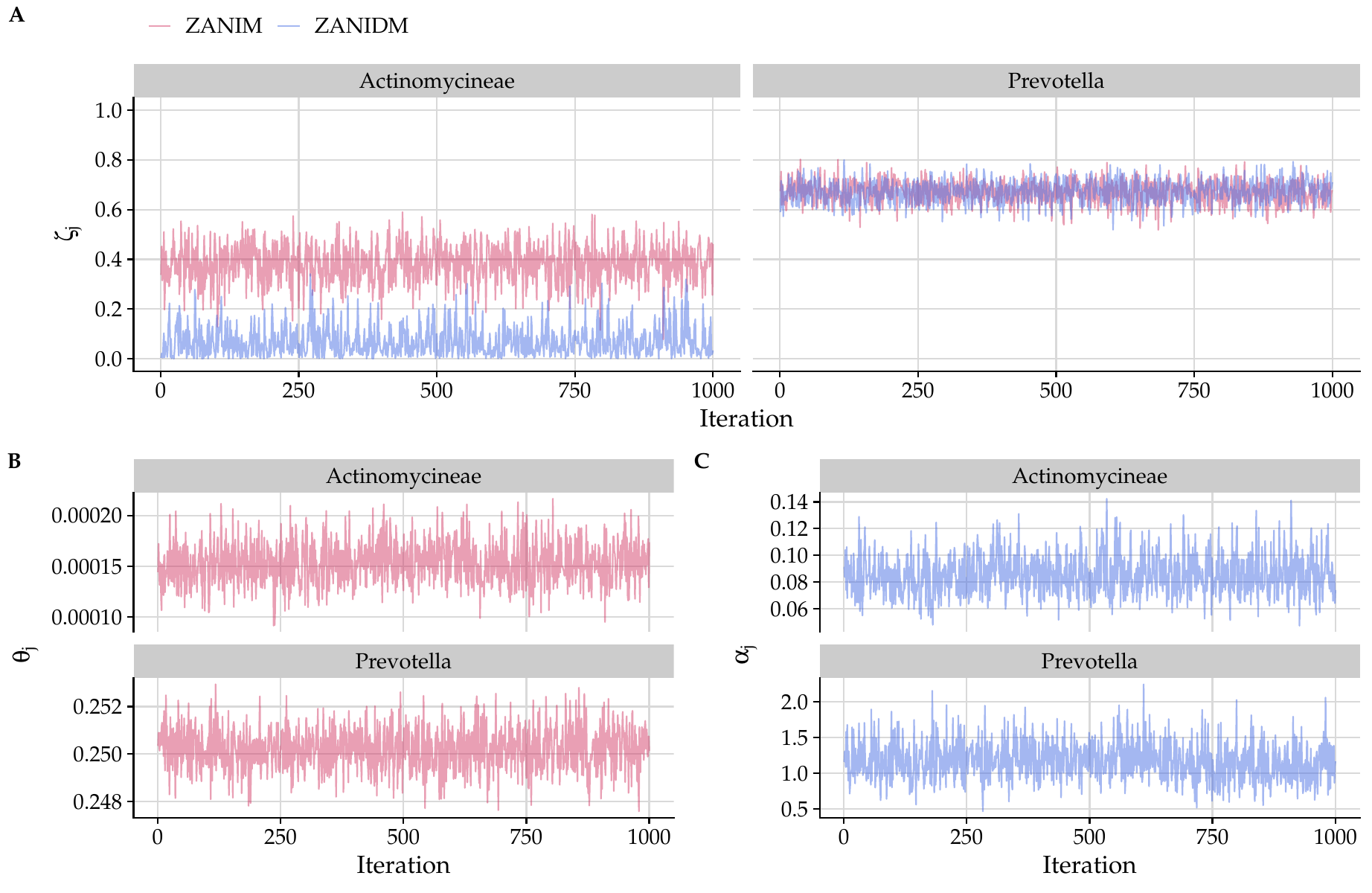}
    \caption{Trace plots for the parameters of the ZANIM and ZANIDM models for the \textit{Actinomycineae} and \textit{Prevotella} taxa. \textbf{A}: $\zeta_j$ under both models. \textbf{B}: ZANIM's success probability parameter $\theta_j$. \textbf{C}: ZANIDM's concentration parameter $\alpha_j$.}
    \label{fig:posterior_chosen_taxa}
\end{figure}

For the \textit{Actinomycinae} taxon depicted in \autoref{fig:posterior_chosen_taxa}, it is notable that the posterior mean of the zero-inflation parameter $\zeta_j$ differs significantly between the two models, with $\zeta_j$ being notably higher under the ZANIM model (see Panel \textbf{A}). This, coupled with the relatively low posterior mean for $\alpha_j$ in Panel \textbf{C}, highlights the extent to which the DM components in ZANIDM can account for the observed zeros ($56$ out of $98$ observations, for this taxon) through overdispersion. Conversely, under the ZANIM model, the low posterior mean for $\theta_j$  (see Panel \textbf{B}), coupled with a higher posterior mean for $\zeta_j$, shows that the zeros are largely attributed to structural zero-inflation by ZANIM.

The \textit{Prevotella} taxon, also depicted in \autoref{fig:posterior_chosen_taxa}, further illustrates the importance of modelling structural zeros in these data. For this taxon, $66$ of the $98$ observations are zeros, which corresponds to roughly $67\%$ of the individuals. According to Panel \textbf{A}, the ZANIM and ZANIDM models both appear to largely attribute this to structural zero-inflation, since the posterior distribution of $\zeta_j$ is concentrated around $0.67$ in each case.

The remaining trace plots in the Supplementary Material show similar behaviour, in that the posterior mean of the $\zeta_j$ parameter is consistently no lower under the ZANIDM model than under the ZANIM model. This is consistent with the aforementioned overdispersion exhibited by all taxa. In light of this, it is not surprising that the ZANIDM model provides a better fit to these data, with its estimated ELPD being significantly higher than that of the ZANIM model. Overall, although covariates are also available for these data, and accounting for them within our framework could improve the analysis, the results for both models clearly highlight the need to account for the co-occurrence of zero-inflation across taxa in these data.
 \section{Discussion}\label{sec:discussion}

The main contributions of this paper have been the novel probabilistic insights provided for the ZANIM and ZANIDM
distributions, which are suitable for addressing zero-inflation in count-compositional data.
We provided a proper probabilistic characterisation of the ZANIDM distribution, first introduced by
\citet{Koslovsky2023}, and developed the more parsimonious ZANIM distribution.
We demonstrated that both distributions belong to a unifying framework and can be represented as finite mixtures.
We derived their key properties, including moments and the corresponding marginal distributions.
We showed that the distributions can accommodate overdispersion and positive correlations, which can be attributed
to their mixture structure and zero-inflation properties.

We subsequently developed Bayesian inference frameworks for both distributions. Specifically, for ZANIDM, we
showed through simulation studies that our collapsed Gibbs sampling approach for updating the latent parameters is more efficient than the
algorithm of \citet{Koslovsky2023}. Our extensive simulations and our application to a human gut microbiome dataset also showed that both distributions are effective when
data exhibit zero-inflation across multiple categories.
It is worth noting that if exclusively non-zero counts are observed in one or more categories, or if contextual information gives sufficient reason to believe that the observed zeros for a given $Y_j$ are not structural in nature, it is possible to simplify the distributions by removing the corresponding zero-inflation parameters $\zeta_j$ from the model.

Our unifying framework that characterises both distributions as finite mixtures can be extended to incorporate other component distributions suitable for modelling count-compositional data, e.g., the Conway-Maxwell-multinomial distribution \citep{Kadane2018,Morris2020}.
In doing so, novel zero-and-$N$-inflated counterparts for such distributions could be developed, and theoretical insights similar to those for the ZANIM and ZANIDM distributions could be provided, though inference would remain a challenge.
It is also of interest to extend our framework to incorporate compositional component distributions which handle proportions, such as the Dirichlet distribution, for settings where information is available on relative abundances rather than raw counts.

A general limitation of the proposed ZANIM and ZANIDM distributions is our assumption of independence for the latent binary indicators $\mathbf{z}$ under both distributions.
This is explicitly reflected in their respective stochastic representations in \autoref{prop:stochastic_representation_zanim} and \autoref{def:zanidm_stochastic_representation}, which both proceed from the assumption that $(z_j\mid\zeta_j)\sim\operatorname{Bernouilli}\lbrack1-\zeta_j\rbrack$ are independent for $j \in \{1, \ldots, d\}$. This assumption may be unrealistic in some real data settings, as zeros may be liable to co-occur in two or more categories because of dependencies or shared latent factors.
In the context of a multivariate zero-inflated Poisson model, \citet{Lee2018} address this issue by assuming a multivariate normal prior for the latent binary indicators and use a multivariate probit regression model. An alternative way to relax this assumption in our setting relies on the fact that the augmented likelihoods for both distributions factor into a product of independent, category-specific terms --- where $\zeta_j$ is in turn conditionally independent of the $\lambda_j$ and $\alpha_j$ parameters of the ZANIM and ZANIDM distributions, respectively --- and the likelihood contribution for $\zeta_j$ has a Bernoulli form. Thus, it is possible to introduce shared latent factors across the categories via
$g(\zeta_j)=\beta_j + u_j$, with $\mathbf{u} \sim \operatorname{N}_d\lbrack\mathbf{0}_d, \mathbf{R}\rbrack$,
where $g(\cdot)$ is a suitable link function and $\mathbf{R}$ is a non-diagonal covariance matrix.
Extending our approach in this direction is of great interest and may improve the practical utility of our proposed models. We stress that the theoretical properties derived for both distributions would continue to hold, conditional on the shared latent factors, although the inference procedures would be more involved.

A related future research direction would be to explore non-parametric regression approaches for incorporating covariates into both distributions' category-specific parameters. Indeed, covariates are available for the
microbiome dataset presented in Section \ref{sec:microbiome} and accounting for them within our framework would improve the analyses. For the ZANIM model, this would extend the approach of \citet{Zeng2023}, who allow only the success probabilities, and not the $\zeta_{j}$ parameters, to depend on covariates. For the ZANIDM model, this would provide additional flexibility over the log-linear and logistic regressions employed by \citet{Koslovsky2023} for the $\alpha_{j}$ and $\zeta_{j}$ parameters, respectively. We also stress that, even in regression settings, our Gibbs updates of $\lambda_{ij}$ and $z_{ij}$ under ZANIDM would still be advantageous.

Overall, we envisage that the novel theoretical insights we provide for the ZANIM and ZANIDM distributions will be of interest to researchers and applied practitioners working with either distribution or with zero-inflated multivariate data more broadly.

\section*{Acknowledgments}
{\small
This publication has emanated from research conducted with the financial support of Taighde {\'{E}}ireann -- Research Ireland under Grant number 18/CRT/6049.
Andrew Parnell’s work was supported by: the UCD-Met {\'{E}}ireann Research Professorship Programme (28-UCDNWPAI); a Research Ireland Research Centre award 12/RC/2289\_P2; the Research Ireland Centre for Research Training 18/CRT/6049; and Research Ireland Co-Centre Climate+ in Climate Biodiversity and Water award 22/CC/11103. For the purpose of Open Access, the author has applied a CC BY public copyright licence to any Author Accepted Manuscript version arising from this submission.}

\appendix
\section*{Appendices}
\renewcommand{\thetable}{\Alph{section}.\arabic{table}}
\renewcommand{\thefigure}{\Alph{section}.\arabic{figure}}
\renewcommand{\theequation}{\Alph{section}.\arabic{equation}}
\renewcommand{\thealgocf}{\Alph{section}.\arabic{algocf}}
\counterwithin*{equation}{section}
\counterwithin*{table}{section}
\counterwithin*{figure}{section}
\counterwithin*{algocf}{section}
\SetFuncSty{textit}

\section{Derivation of the ZANIM PMF}\label{app:zanim_pmf}
\begin{proof}[\textbf{\upshape Proof of \autoref{theo:zanim_pmf}:}]

The goal is to marginalise out the latent variable $\phi$ from \eqref{eq:pmf_zanim_augmented}
and ensure that the function $\int p(\mathbf{y}, \phi; \bm{\lambda}, \bm{\zeta}) \dd \phi$
will be a proper PMF.
We shall denote $p(\mathbf{y}, \phi; \bm{\lambda}, \bm{\zeta})=p(\mathbf{y}, \phi)$, for brevity.
We begin with
\[
p(\mathbf{y}, \phi)= \dfrac{N!\phi^{N - 1}}{\Gamma(N)}\prod_{j=1}^d \left\lbrack\zeta_j\mathds{1}_0(y_{j}) + (1 - \zeta_j)\dfrac{\lambda_j^{y_{j}}e^{-\lambda_j\phi}}{y_{j}!}\right\rbrack,
\]
and note that the product will have $2^d$ terms, as a result of the binomial theorem.
However, due to the indicator function $\mathds{1}_0(y_{j})$, we can simplify some terms.
We shall consider four groups of terms, corresponding to the four types of mixture component in the ZANIM PMF.

$\bullet$ Standard multinomial$\colon$
\begin{align*}
\int p(\mathbf{y}, \phi)\dd \phi &=
\displaystyle\binom{N}{y_{1} \dots y_{d}}\prod_{j=1}^d(1-\zeta_j)
\left(\dfrac{\lambda_j}{\sum_{k=1}^d\lambda_k}\right)^{y_{j}} =\displaystyle\binom{N}{y_{1} \dots y_{d}}\prod_{j=1}^d(1-\zeta_j)\theta_j^{y_{j}}.
\end{align*}

$\bullet$ $\delta_{\mathbf{0}_d}$ component$\colon$
\begin{align*}
\int p(\mathbf{y}, \phi)\dd \phi &=
\left\lbrack\prod_{j=1}^d\zeta_j\mathds{1}_0(y_{j})\right\rbrack N!
\int \dfrac{1}{\Gamma(N)}\phi^{N-1}\dd \phi=
\left\lbrack\prod_{j=1}^d\zeta_j\mathds{1}_0(y_{j})\right\rbrack N!0^{-N}
\int \dfrac{0^{N}}{\Gamma(N)}\phi^{N-1}e^{-0\phi}\dd \phi\\
&=\prod_{j=1}^d\zeta_j\mathds{1}_0(y_{j}).
\end{align*}
The simplification arises due to the fact that $N = 0$ when $y_j=0\:\forall\:j$, such that  $N!=0!=1$ and $0^0=1$, by convention. Consequently, the integral above can be seen as an abstract representation of a $\operatorname{Gamma}\lbrack 0, 0\rbrack$ distribution,
which in practice is not well-defined.
Under the convention that $\operatorname{Gamma}\lbrack 0, 0\rbrack \overset{d}{=} \delta_0(\cdot)$, the integral evaluates to $1$.

$\bullet$ N$\mathbf{e}_d^{(j)}$ components$\colon$

We have $d$ terms with this constraint, which can be written as follows:
\[
\int p(\mathbf{y}, \phi)\dd \phi =
\sum_{j=1}^d
\left\{\mathds{1}_0\left(\sum_{k\colon k\neq j} y_k \right)(1 - \zeta_j)\prod_{k\colon k\neq j} \zeta_k
\int \dfrac{N!\phi^{N - 1}}{\Gamma(N)}\dfrac{\lambda_j^{y_j}e^{-\lambda_j\phi}}{y_j!}\dd\phi\right\}=
\sum_{j=1}^d
\left\{\mathds{1}_0\left(\sum_{k\colon k\neq j} y_k \right)(1 - \zeta_j)\prod_{k\colon k\neq j} \zeta_k\right\},
\]
where the simplification comes from the fact that $N=y_j$
when $y_j > 0$ and $y_k = 0$ $\forall \: k \neq j$.

$\bullet$ Reduced multinomials$\colon$

The remaining terms represent cases where at most $d-2$ categories exhibit zero-inflation.
To write the $2^{d}-d-2$ such terms compactly, we recall the corresponding set
$\mathfrak{K}= \{\mathcal{K} \subseteq \{1,\ldots,d\}; 1 \leq \lvert\mathcal{K}\rvert \leq d-2\}$
described in \autoref{def:set_K}.
Using this, we can write
\[
p(\mathbf{y}, \phi) = \dfrac{N!\phi^{N - 1}}{\Gamma(N)}\sum_{\mathcal{K} \in \mathfrak{K}}
\left\{\mathds{1}_0\left(\sum_{k \in \mathcal{K}} y_k \right)
\prod_{k \in \mathcal{K}}\zeta_k\prod_{j \notin \mathcal{K}}
\left\lbrack (1-\zeta_j)\dfrac{\lambda_j^{y_j}e^{-\lambda_j\phi}}{y_j!} \right\rbrack\right\}.
\]
Then, for a generic set $\mathcal{K}$, we have
\begin{align*}
\int p(\mathbf{y}, \phi)\dd\phi &= \mathds{1}_0\left(\sum_{k \in \mathcal{K}} y_k \right)\prod_{k \in \mathcal{K}}\zeta_k
\prod_{j \notin \mathcal{K}}\left\lbrack (1-\zeta_j)\dfrac{\lambda_j^{y_j}}{y_j!} \right\rbrack
N!\int \dfrac{e^{-\phi\sum_{j \notin \mathcal{K}}\lambda_j}\phi^{N - 1}}{\Gamma(N)}\dd\phi \\
&=\mathds{1}_0\left(\sum_{k \in \mathcal{K}} y_k \right)\prod_{k \in \mathcal{K}}\zeta_k
\prod_{j \notin \mathcal{K}}\left\lbrack (1-\zeta_j)\right\rbrack
\binom{N}{\{y_j\}_{j \notin \mathcal{K}}}
\prod_{j \notin \mathcal{K}}\left\lbrack \left(\theta_j^{\mathcal{K}}\right)^{y_j}\right\rbrack,
\end{align*}
where
$\theta_j^{\mathcal{K}} = \lambda_j/\sum_{k \notin \mathcal{K}} \lambda_k =
\theta_j/(1 - \sum_{k \in \mathcal{K}} \theta_k)$.
The simplification arises from the fact that we have $y_k = 0$ for the indices $k \in \mathcal{K}$, such that $N = \sum_{j=1}^dy_j = \sum_{j \notin \mathcal{K}} y_j$.

Collecting the terms and using \autoref{def:etas} leads to \autoref{theo:zanim_pmf}.
\end{proof} \section{Derivation of the ZANIDM PMF}\label{app:zanidm_pmf}

\begin{proof}[\textbf{\upshape Proof of \autoref{theo:zanidm_pmf}:}]
From the stochastic representation in \autoref{def:zanidm_stochastic_representation},
we highlight a key fact that will be used extensively below,
which is that $\lambda_j=0$ (i.e., $z_{j}=0$) implies $y_j = 0$.
Next, we note that we can marginalise out $z_{j}$ in the distribution of the latent $\lambda_{j}\in\{0,\mathbb{R}^{+}\}$ via
\[
p(\lambda_{j} \mid \alpha_{j}, \zeta_j)
= \sum_{k=0}^1 p(\lambda_{j} \mid z_{j} = k, \alpha_{j})p(z_{j} = k \mid \zeta_j) =
\zeta_j\delta_0(\lambda_{j}) + (1-\zeta_j) \frac{\lambda_{j}^{\alpha_j - 1} e^{-\lambda_{j}}}{\Gamma(\alpha_j)}.
\]
Consequently, we note that
$(\lambda_{j} \mid \alpha_j, \zeta_j) \sim \operatorname{ZAG}\lbrack 1-\zeta_j, \alpha_{j}, 1\rbrack$,
i.e., unconditional on $z_{j}$, the latent variable $\lambda_{j}$ follows a zero-augmented gamma
distribution with shape $\alpha_j$, rate $1$, and $\zeta_j$ being the probability that $\lambda_{j}=0$. Thus, the probability density function of $(\lambda_{j} \mid \alpha_j, \zeta_j)$ is given by

\[
p(\lambda_j \mid \zeta_j, \alpha_j) =
\zeta_j^{\mathds{1}_0(\lambda_j)}(1 - \zeta_j)^{1 - \mathds{1}_0(\lambda_j)}
\left(
\frac{\lambda_j^{\alpha_j - 1}e^{-\lambda_j}}{\Gamma(\alpha_j)}
\right)^{1 - \mathds{1}_0(\lambda_j)}, \quad \lambda_j \in \{0,\mathbb{R}^{+}\}.
\]

Note that the augmented likelihood for the ZANIDM distribution can be written as
\begin{align*}
    \mathcal{L}(\mathbf{y},\bm{\lambda}; \bm{\alpha},\bm{\zeta})
    &=
    \binom{N}{y_1, \ldots, y_d}\prod_{j=1}^d\left(\frac{\lambda_j}{\sum_{k=1}^d\lambda_k}\right)^{y_j}
    \left\lbrack
    \zeta_j^{\mathds{1}_0(\lambda_j)} (1-\zeta_j)^{1 - \mathds{1}_0(\lambda_j)}
    \left(\frac{\lambda_j^{\alpha_j - 1}e^{-\lambda_j}}{\Gamma(\alpha_j)}\right)^{1 - \mathds{1}_0(\lambda_j)}
    \right\rbrack \\
&=
    c
    \prod_{j=1}^d
    \left\lbrack
    \zeta_j^{\mathds{1}_0(\lambda_j)}(1-\zeta_j)^{1 - \mathds{1}_0(\lambda_j)}
    \right\rbrack
    \prod_{j=1}^d
    \left\lbrack
    \left(\frac{1}{\sum_{k=1}^d\lambda_k}\right)^{y_j}
    \left(\frac{\lambda_j^{y_j + \alpha_j - 1}e^{-\lambda_j}}{\Gamma(\alpha_j)}\right)
    \right\rbrack^{1 - \mathds{1}_0(\lambda_j)},
\end{align*}
where we denote the constant $c = \binom{N}{y_1, \ldots, y_d}=\Gamma(N + 1) / \prod_{j=1}^d \Gamma(y_j + 1)$, for brevity, and the simplification relies on $\theta_j^{y_j}$ being non-zero when $\lambda_j > 0$.
In light of this, the marginal PMF of $\mathbf{Y}$ is obtained by integrating out the latent variables $\lambda_j$:
\[
\Pr\left\lbrack \mathbf{Y} = \mathbf{y}; \bm{\alpha},\bm{\zeta} \right\rbrack =
\int\cdots\int \mathcal{L}(\mathbf{y},\bm{\lambda};\bm{\alpha},\bm{\zeta})
\dd\lambda_1\ldots\dd\lambda_d = \int_{\mathbb{R}^d} \mathcal{L}(\mathbf{y},\bm{\lambda};\bm{\alpha},\bm{\zeta})\dd\bm{\lambda}
\]
Before we proceed with the integration, we state the integral result;
consider the change of variables $s = \sum_{j=1}^d\lambda_j$ and $\lambda_j = \theta_j s$, which leads to $\dd\lambda_1\ldots \dd \lambda_d = s^{d-1}\dd s\prod_{j=1}^{d-1}\dd\theta_j$. Note that the vector $\bm{\theta} = (\theta_1, \ldots \theta_d)$ belongs to the simplex $\mathcal{S}^d = \{\bm{\theta}: \theta_j > 0, j \in \{1,\ldots, d\}, \sum_{j=1}^d\theta_j = 1\}$, which leads to the following multivariate beta integral result:
\[
\int_{\mathcal{S}^d} \prod_{j=1}^d \theta_j^{k_j-1} \dd \bm{\theta} =
\int_{0}^{1}\int_{0}^{1- \theta_1} \cdots\int_{0}^{1-\theta_1-\ldots- \theta_{d-2}}
\theta_1^{k_1-1}\cdots \theta_{d}^{k_d-1} \dd \theta_1 \ldots \dd \theta_{d-1} =
\frac{\prod_{j=1}^d\Gamma(k_j)}{\Gamma\left(\sum_{j=1}^dk_j\right)},
\]
where we do not need to integrate with respect to $\theta_d$ explicitly, since $\theta_d$ is fully determined by the simplex constraint $\sum_{j=1}^d\theta_j=1$.

Since $\lambda_{j} \in \{0, \mathbb{R}^{+}\}\:\forall\:j \in \{1,\ldots, d\}$, the integration should be performed considering all $2^d$ combinations of $\lambda_{j}$ being $0$ or non-zero. We consider four different groups of terms, as per  \ref{app:zanim_pmf}, and adopt the notation $\alpha_s = \sum_{j=1}^d\alpha_j$ and $\alpha^{\mathcal{K}}_s = \sum_{j \notin \mathcal{K}} \alpha_j$.

$\bullet$ No inflation$\colon$
When $\lambda_j>0 \: \forall\:j$, the integral becomes
\begin{align*}
\int_{\mathbb{R}^d} \mathcal{L}(\mathbf{y},\bm{\lambda};\bm{\alpha},\bm{\zeta})\dd\bm{\lambda} &=
c\prod_{j=1}^d(1-\zeta_j)\int\cdots\int
\prod_{j=1}^d
\left\lbrack
\left(\frac{1}{\sum_{k=1}^d\lambda_k}\right)^{y_j}
\left(\frac{\lambda_j^{y_j + \alpha_j - 1}e^{-\lambda_j}}{\Gamma(\alpha_j)}\right)
\right\rbrack
\dd \lambda_1\ldots\dd\lambda_d\\
&=\frac{\Gamma(N + 1)\Gamma(\alpha_s)}{\Gamma(N + \alpha_s)}
\prod_{j=1}^d(1-\zeta_j)\frac{\Gamma(y_j + \alpha_j)}{\Gamma(\alpha_j)\Gamma(y_j + 1)}.
\end{align*}

$\bullet$ `All'-inflation$\colon$ When $\lambda_{j} = 0 \: \forall\:j$, then $y_j = 0 \: \forall\:j$ and $c=1$, such that the integral becomes
\[
\int_{\mathbb{R}^d} \mathcal{L}(\mathbf{y},\bm{\lambda};\bm{\alpha},\bm{\zeta})\dd\bm{\lambda}
=\prod_{j=1}^d\zeta_j.
\]

$\bullet$ $N$-inflation$\colon$ When $\lambda_{k} = 0 \: \forall\:k \neq j$, we know that $y_k = 0 \: \forall\: k\neq j$, such that $c = 1$, since $N = y_j$.
Note that we have $d$ terms with this constraint. We can write these terms as follows

\[
\int_{\mathbb{R}^d} \mathcal{L}(\mathbf{y},\bm{\lambda};\bm{\alpha},\bm{\zeta})\dd\bm{\lambda}=
c\prod_{k\neq j}\zeta_k(1 - \zeta_j)\int \frac{\lambda_j^{\alpha_j - 1}e^{-\lambda_j}}{\Gamma(\alpha_j)}\dd \lambda_j =(1 - \zeta_j)\prod_{k\neq j}\zeta_k.
\]

$\bullet$ Sets of inflation$\colon$
As per the derivation of the ZANIM PMF in \ref{app:zanim_pmf},
the remaining terms represent cases where at most $d-2$ categories exhibit zero-inflation and can be written compactly using the set notation introduced in \autoref{def:set_K}. For a given $\mathcal{K} \in \mathfrak{K}$, we know when $\lambda_k = 0$ for all $k \in \mathcal{K}$ that $y_k = 0 \: \forall\:k \in \mathcal{K}$ and that $\lambda_j = 0$ for all $j \notin \mathcal{K}$. Hence, we have
\begin{align*}
\int_{\mathbb{R}^d} \mathcal{L}(\mathbf{y},\bm{\lambda};\bm{\alpha},\bm{\zeta})\dd\bm{\lambda}
&=
c\prod_{k \in \mathcal{K}}\zeta_k
\prod_{j \notin \mathcal{K}}(1-\zeta_j)
\int_{\mathbb{R}^{\lvert\mathcal{K}\rvert}}
\prod_{j \notin \mathcal{K}}
\left\lbrack
\left(\frac{1}{\sum_{\ell \notin \mathcal{K}}\lambda_\ell}\right)^{y_j}
\frac{\lambda_j^{y_j + \alpha_j - 1}e^{-\lambda_j}}{\Gamma(\alpha_j)}
\right\rbrack\bm{\dd}\bm{\lambda}^{(\mathcal{K})} \\
&=
\frac{\Gamma(\alpha^{\mathcal{K}}_s)\Gamma(N+1)}{\Gamma(N + \alpha^{\mathcal{K}}_s)}
\prod_{k \in \mathcal{K}}\zeta_h
\prod_{j \notin \mathcal{K}}(1-\zeta_j) \frac{\Gamma(y_j + \alpha_j)}{\Gamma(\alpha_j)\Gamma(y_j + 1)},
\end{align*}
where the multivariate integral is over the set $\bm{\lambda}^{(\mathcal{K})}=\{\lambda_j\colon \lambda_j \notin \mathcal{K}\}$.

Collecting the terms while accounting for the restriction that $y_j = 0$ when $\lambda_j = 0$
and using \autoref{def:etas} leads to \autoref{theo:zanidm_pmf}.
\end{proof}
\bibliographystyle{myjmva}
\bibliography{references}

\section*{Supplementary Material}
\renewcommand{\thetable}{S.\arabic{table}}
\renewcommand{\thefigure}{S.\arabic{figure}}
\renewcommand{\thesubsection}{S.\arabic{subsection}}
\renewcommand{\theequation}{S.\arabic{equation}}
\renewcommand{\thealgocf}{S.\arabic{algocf}}
\renewcommand{\thesubsection}{Supp. Mat.~\Alph{subsection}}

In the Supplementary Material,
we present detailed derivations of the ZANIM and ZANIDM inference schemes in
\ref{app:bernoulli_gamma_da} and \ref{app:zanidm_inference_details}, respectively,
the derivation of the moment generating functions via mixture properties in \ref{app:mgf},
posterior summaries for alternative ZANIDM inference schemes in \ref{app:posterior_summaries_zanidm},
additional simulation results with balanced parameter settings in \ref{app:balanced_experiments},
and additional results for the microbiome data analysis in \ref{app:add_microbiome}.

\subsection{ZANIM inference via data augmentation}\label{app:bernoulli_gamma_da}

Inference for the ZANIM parameters is based on the augmented likelihood $p(\mathbf{y}, \phi, \mathbf{z}) = p(\phi \mid \mathbf{y}, \mathbf{z})p(\mathbf{y} \mid \mathbf{z})p(\mathbf{z})$. We establish the validity of this approach by showing how the augmented likelihood in  \eqref{eq:pmf_zanim_augmented}, from which the ZANIM distribution was initially derived, can be recovered from this expression. For this derivation, we adopt the notation $c = \binom{N}{y_1\cdots y_d}$ and drop the subscript $i$, for simplicity.

Let $\mathbf{z} = (z_1, \ldots, z_d)$, where $z_j=0$ corresponds to a structural zero count and $z_j = 1$ represents a count obtained from a sampling distribution (which may also be zero). Assuming $z_j \sim \operatorname{Bernoulli}\left\lbrack 1 - \zeta_j \right\rbrack$ and independence over $j\in\{1,\ldots,d\}$, we have
$p(\mathbf{z}) = \prod_{j=1}^d (1-\zeta_j)^{z_j}\zeta_j^{1-z_j}$.
Conditional on $\mathbf{z}$, we can fully determine which one of the $2^d$ mixture components
from ZANIM that $\mathbf{y}$ belongs to.
This is important because we do not need to use the common mixture model data augmentation
which requires $K=2^d$ latent variables.
Instead, we introduce a latent variable, conditional on $\mathbf{z}$, given by
$(\phi \mid \mathbf{y}, \mathbf{z}) \sim \operatorname{Gamma}\lbrack N, \sum_{j=1}^d \lambda_jz_j\rbrack$.
This is similar to the data augmentation given in  \eqref{eq:gamma_trick},
though here the contributions of structural zeros are removed from the calculation of the rate parameter.
As per \ref{app:zanim_pmf}, we consider
four groups of terms, corresponding to the four types of mixture component in the ZANIM PMF.

$\bullet$ Standard multinomial component$\colon$
For this component, we have that $\mathbf{z}=\mathbf{1}_d$ and
\begin{align*}
&c\prod_{j=1}^d \left\{\left(\dfrac{\lambda_j}{\sum_{k=1}^d\lambda_k} \right)^{y_j}\right\}\dfrac{\phi^{N-1}}{\Gamma(N)}
\exp\left\lbrack -\phi \sum_{j=1}^d\lambda_jz_j\right\rbrack\left(\sum_{j=1}^d\lambda_jz_j\right)^N
\prod_{j=1}^d\Pr\lbrack z_j=1\rbrack \mathds{1}(z_j=1)\\
&=c\dfrac{\phi^{N-1}}{\Gamma(N)}\prod_{j=1}^d\left\{(1-\zeta_j) \lambda_j^{y_j}e^{-\phi\lambda_j}\mathds{1}(z_j=1)\right\},
\end{align*}
where some simplification arises from the fact that $N=\sum_{j=1}^d y_j$.

$\bullet$ $\delta_{\mathbf{0}_d}$ component$\colon$
For this component, we have  $\mathbf{z}=\mathbf{0}_d$ and, subject to some simplifications,
\[
c\prod_{j=1}^{d}\mathds{1}_0(y_j) \times \dfrac{\phi^{N-1}}{\Gamma(N)}e^{-0\phi}0^{0}\times \Pr\lbrack z_j=0\rbrack \mathds{1}_0(z_j)=c\dfrac{\phi^{N-1}}{\Gamma(N)}\prod_{j=1}^{d}\zeta_j\mathds{1}_0(y_j)\mathds{1}_0(z_j).
\]

$\bullet$ $N\mathbf{e}_d^{(j)}$ components$\colon$
For these components, the vector $\mathbf{z}$ has the value $1$ in one entry only.
Suppose the $j$-th entry is $1$, such that $N=y_j$ and $y_k=0\:\forall\:k\neq j$.
We then have
\begin{align*}
&c\prod_{k\colon k\neq j}\left\{\mathds{1}_0(y_k)\times \dfrac{\phi^{N-1}}{\Gamma(N)}
\exp\left\lbrack -\phi \sum_{j=1}^d\lambda_jz_j
\right\rbrack\left(\sum_{j=1}^d\lambda_jz_j\right)^N\times
\Pr\lbrack z_k=0\rbrack \mathds{1}_0(z_k)
\Pr\lbrack z_j=1\rbrack \mathds{1}(z_j=1)
\right\}\\
&=c\dfrac{\phi^{N-1}}{\Gamma(N)}\prod_{k\colon k\neq j}\left\{\zeta_k \mathds{1}_0(y_k)\mathds{1}_0(z_k) \right\}\times (1 - \zeta_j)\lambda_j^{y_j}e^{-\phi \lambda_j}\mathds{1}(z_j=1).
\end{align*}

$\bullet$ Reduced multinomial components$\colon$
For these components, the vector $\mathbf{z}$ contains $0$ at the entries
$k \in \mathcal{K}$ and $1$ at the entries $j \notin \mathcal{K}$,
such that $y_k=0\:\forall\:k\in\mathcal{K}$.
We then have
\begin{align*}
&c\prod_{k \in \mathcal{K}} \mathds{1}_0(y_k)\prod_{j \notin \mathcal{K}}\left\lbrack \left(\dfrac{\lambda_j}{\sum_{\ell \notin \mathcal{K}}\lambda_\ell}\right)^{y_j}\right\rbrack
\times \dfrac{\phi^{N-1}}{\Gamma(N)}\exp\left\lbrack -\phi \sum_{j=1}^d\lambda_jz_j\right\rbrack\left(\sum_{j=1}^d\lambda_jz_j\right)^N\times \prod_{k \in \mathcal{K}}\Pr\lbrack z_k=0\rbrack \mathds{1}_0(z_k)
\prod_{j \notin \mathcal{K}}\Pr\lbrack z_j=1\rbrack \mathds{1}(z_j=1) \\
&=c\prod_{k \in \mathcal{K}} \mathds{1}_0(y_k)\mathds{1}_0(z_k)\zeta_k
\prod_{j \notin \mathcal{K}}\left\lbrack(1-\zeta_j) \mathds{1}(z_j=1)\left(\dfrac{\lambda_j}{\sum_{\ell \notin \mathcal{K}}\lambda_\ell}\right)^{y_j}\right\rbrack\times \dfrac{\phi^{N-1}}{\Gamma(N)}\exp\left\lbrack -\phi \sum_{j\notin \mathcal{K}}\lambda_jz_j\right\rbrack\left(\sum_{j\notin\mathcal{K}}\lambda_jz_j\right)^N\\
&=c\dfrac{\phi^{N-1}}{\Gamma(N)}\prod_{k \in \mathcal{K}} \mathds{1}_0(y_k)\mathds{1}_0(z_k)\zeta_k
\prod_{j \notin \mathcal{K}}\left\lbrack(1-\zeta_j) \lambda_j^{y_j}e^{-\phi\lambda_j} \mathds{1}(z_j=1)
\right\rbrack.
\end{align*}
By summing over all terms above, where $c^\prime=(c\phi^{N-1})/\Gamma(N)$ is a common factor, we obtain
\begin{align*}
p(\mathbf{y}, \phi, \mathbf{z}) &=
c^\prime
\prod_{j=1}^d\left\{(1-\zeta_j) \lambda_j^{y_j}e^{-\phi\lambda_j}\mathds{1}(z_j=1)\right\} +
c^\prime\prod_{j=1}^{d}\zeta_j\mathds{1}_0(y_j)\mathds{1}_0(z_j) \\
&\phantom{=}~+
c^\prime\sum_{j=1}^d\left\lbrack
\prod_{k\colon k\neq j}\left\{\zeta_k \mathds{1}_0(y_k)\mathds{1}_0(z_k) \right\}\times (1 - \zeta_j)\lambda_j^{y_j}e^{-\phi \lambda_j}\mathds{1}(z_j=1)\right\rbrack \\
&\phantom{=}~+c^\prime\sum_{\mathcal{K} \in \mathfrak{K}}\left\lbrack
\prod_{k \in \mathcal{K}} \mathds{1}_0(y_k)\mathds{1}_0(z_k)\zeta_k\prod_{j \notin \mathcal{K}}\left\lbrack(1-\zeta_j) \lambda_j^{y_j}e^{-\phi\lambda_j} \mathds{1}(z_j=1)\right\rbrack\right\rbrack.
\end{align*}
We can also factor out the terms
$
\mathds{1}_0(z_{j})\mathds{1}_0(y_{j})\zeta_j + \mathds{1}(z_{j}=1)(1 - \zeta_j)\lambda_j^{y_{j}}e^{-\phi\lambda_j}$. Then, by noting that $\mathds{1}_0(z_{j}) = 1 - z_j$ and $\mathds{1}(z_{j}=1) = z_j$, we can express the
above sum as
\begin{align}
p(\mathbf{y}, \phi, \mathbf{z}) \label{eq:augmented_likelihood_zanim_phi_z_y}
&=c^\prime\prod_{j=1}^d\left\{\left\lbrack(1 - z_j)\zeta_j + z_j(1-\zeta_j)e^{-\phi\lambda_j}\right\rbrack^{\mathds{1}_0(y_j)}
\times\left\lbrack z_j(1-\zeta_j) \lambda_j^{y_j}e^{-\phi\lambda_j}\right\rbrack^{1-\mathds{1}_0(y_j)}\right\}.
\end{align}
Importantly, the augmented likelihood factors into separate terms for each category after conditioning on $\mathbf{z}$ and $\phi$. Furthermore, we note that the likelihood contribution within a given category is a
product of two terms; one for when $y_j = 0$ and one for when $y_j>0$.

To derive $(\mathbf{z} \mid \mathbf{y}, \phi)$, we first note that
$
p(\mathbf{z} \mid \mathbf{y}, \phi) = \prod_{j=1}^d p(z_j \mid y_j, \phi)$,
since the categories are conditionally independent, as seen by  \eqref{eq:augmented_likelihood_zanim_phi_z_y}. For a given category $j$ when $y_j > 0$, we have that
$p(z_j \mid y_j > 0 , \phi) \propto z_j(1-\zeta_j) \lambda_j^{y_j}e^{-\phi\lambda_j}$.
It is evident that $p(z_j = 1\mid y_j >0 , \phi) = 1$, hence $(z_j \mid y_j > 0, \phi)$ is a degenerate distribution at $1$ when $y_j > 0$. On the other hand, when $y_j = 0$, we have that
$p(z_j \mid y_j = 0 , \phi) \propto (1 - z_j)\zeta_j + z_j(1-\zeta_j)e^{-\phi\lambda_j}$.
Since $z_j \in \{0, 1\}$, we obtain
\[
p(z_j \mid y_j = 0 , \phi) = \dfrac{(1 - z_j)\zeta_j + z_j(1-\zeta_j)e^{-\phi\lambda_j}}{\zeta_j + (1-\zeta_j)e^{-\phi\lambda_j}},
\]
and can therefore characterise the distribution of $(z_j \mid y_j, \phi)$ as per  \eqref{eq:bernoulli_augmented}.
Finally, it is easy to see from  \eqref{eq:augmented_likelihood_zanim_phi_z_y} that summing over $\mathbf{z}$ yields the desired result
\begin{align*}
p(\mathbf{y}, \phi) &=c^\prime\prod_{j=1}^d\left\{
\sum_{z_j=0}^1
\left\lbrack(1 - z_j)\zeta_j + z_j(1-\zeta_j)e^{-\phi\lambda_j}\right\rbrack^{\mathds{1}_0(y_j)}
\times\left\lbrack z_j(1-\zeta_j) \lambda_j^{y_j}e^{-\phi\lambda_j}\right\rbrack^{1-\mathds{1}_0(y_j)}\right\} \nonumber\\\nonumber
&=\dfrac{N!\phi^{N-1}}{\Gamma(N)} \prod_{j=1}^d\left\{\zeta_j\mathds{1}_0(y_j)  + (1-\zeta_j) \dfrac{\lambda_j^{y_j}e^{-\phi\lambda_j}}{y_j!} \right\}.
\end{align*}

The overall Bayesian inference scheme under this data augmentation strategy is presented in \autoref{alg:bernoulli_gamma_data_augmentation}.\medskip

\begin{algorithm}[H]
\caption{MCMC inference algorithm for the ZANIM distribution.}
\label{alg:bernoulli_gamma_data_augmentation}
\SetAlgoLined
\SetAlgoItemize
\SetKwInput{KwInput}{Input}
\SetKwInput{KwOutput}{Output}
\SetKwInput{KwInit}{Initialise}
\KwInput{Data $\{\mathbf{y}_i; 1 \leq i \leq n\}$ and number of Monte Carlo iterations $R_{\mathrm{MCMC}}$.}
\KwInit{$\phi^{(0)}_i$, $\lambda^{(0)}_j$, and $z^{(0)}_{ij}$ for $i \in\{1,\ldots, n\}$ and $j \in \{1,\ldots, d\}$.}
\For{iterations $t$ from $1$ to $R_{\mathrm{MCMC}}$}{
    \For{categories $j$ from $1$ to $d$}{
        \begin{enumerate}
            \item Sample $(\zeta^{(t)}_j \mid  \mathbf{y}, \bm{\phi}^{(t-1)}, \mathbf{z}^{(t-1)})$ from its full conditional in \eqref{eq:posterior_zeta}.
            \item Sample $(\lambda^{(t)}_j \mid \mathbf{y}, \bm{\phi}^{(t-1)}, \mathbf{z}^{(t-1)})$ from its full conditional in \eqref{eq:posterior_lambda}.
            \item Update $(z_{ij}^{(t)} \mid y_{ij}, \phi_i^{(t-1)})$ from its full conditional in  \eqref{eq:bernoulli_augmented}.
        \end{enumerate}}$\bullet$ Update $(\phi^{(t)}_i \mid \mathbf{y}_i, \mathbf{z}_{i}^{(t)}) \sim \operatorname{Gamma}\left\lbrack N_i, \sum_{j=1}^{d}\lambda^{(t)}_jz_{ij}^{(t)}\right\rbrack, \quad \mathrm{for} \quad i\in\{1, \ldots, n\}$.
}
\end{algorithm}
 \subsection{ZANIDM inference via data augmentation}\label{app:zanidm_inference_details}
Here, we provide more details on the derivations presented in Section \ref{sec:zanidm_inference}. First, note that the probability density function of $\lambda_{ij}$ can be written as
\[
p(\lambda_{ij} \mid z_{ij}) =(1 - z_{ij})^{\mathds{1}_0(\lambda_{ij})}z_{ij}^{1 - \mathds{1}_0(\lambda_{ij})} \left( \dfrac{\lambda_{ij}^{\alpha_j - 1}e^{-\lambda_{ij}}}{\Gamma(\alpha_j)}\right)^{1 - \mathds{1}_0(\lambda_{ij})}.
\]
Clearly, for $z_{ij}$, we have $p(z_{ij}) = (1 - \zeta_j)^{z_{ij}}\zeta_{j}^{1 - z_{ij}}$.
Then, from the augmented likelihood given in Section \ref{sec:zanidm_inference}, the full joint distribution of $\lambda_{ij}$ and $z_{ij}$ given the observed data $y_{ij}$ is
\begin{align*}
p(\lambda_{ij}, z_{ij} \mid y_{ij}, \phi_i) &\propto
(1 - \zeta_j)^{z_{ij}}\zeta_{j}^{1 - z_{ij}}
(1 - z_{ij})^{\mathds{1}_0(\lambda_{ij})}z_{ij}^{1 - \mathds{1}_0(\lambda_{ij})}
\times \left(\dfrac{\lambda_{ij}^{y_{ij}+\alpha_j - 1}e^{-\lambda_{ij}(1 + \phi_i) }}{\Gamma(\alpha_j)}\right)^{1 - \mathds{1}_0(\lambda_{ij})} \\
&\propto (1 - \zeta_j)^{z_{ij}}\zeta_{j}^{1 - z_{ij}}
\left\lbrack (1 - z_{ij})\delta_0(\lambda_{ij}) + z_{ij} \dfrac{\lambda_{ij}^{y_{ij} + \alpha_j - 1}e^{-\lambda_{ij}(1 + \phi_i) }}{\Gamma(\alpha_j)} \right\rbrack.
\end{align*}
The marginalisation of the joint distribution $p(\lambda_{ij}, z_{ij} \mid y_{ij}, \phi_i)$ with respect to
$\lambda_{ij}$ is given by
\begin{align*}
p(z_{ij} \mid y_{ij}, \phi_i) &=\int p(\lambda_{ij}, z_{ij} \mid y_{ij}, \phi_i) \dd \lambda_{ij} \\
&\propto (1 - \zeta_j)^{z_{ij}}\zeta_{j}^{1 - z_{ij}}
\left\lbrack (1 - z_{ij})\int \delta_0(\lambda_{ij}) \dd \lambda_{ij} +
z_{ij} \int \dfrac{\lambda_{ij}^{y_{ij} + \alpha_j - 1}e^{-\lambda_{ij}(1 + \phi_i) }}{\Gamma(\alpha_j)} \dd \lambda_{ij} \right\rbrack \\
&\propto(1 - z_{ij})\zeta_{j} + z_{ij}(1 - \zeta_j)\dfrac{(1 +\phi_i)^{-(y_{ij} + \alpha_j)}\Gamma(y_{ij} + \alpha_j)}{\Gamma(\alpha_j)}.\end{align*}
When $y_{ij} > 0$, we know that $z_{ij}=1$ almost surely.
Conversely, when $y_{ij} = 0$, we have that
$p(z_{ij} \mid y_{ij}=0, \phi_i) \propto \left\lbrack(1 - \zeta_j) (1 +\phi_i)^{-\alpha_j}\right\rbrack^{z_{ij}}
\zeta_{j}^{1 - z_{ij}}$. Since $z_{ij} \in \{0, 1\}$, it is easy to obtain the normalising constant and write the distribution of $(z_{ij} \mid y_{ij}, \phi_i)$ as per  \eqref{eq:zanidm_z_update}.
On the other hand, the distribution of $(\lambda_{ij} \mid  y_{ij}, z_{ij}, \phi_i)$, which yields  \eqref{eq:zanidm_lambda_update} when normalised, is
\[
p(\lambda_{ij} \mid y_{ij}, z_{ij}, \phi_i) \propto (1 - z_{ij})\delta_0(\lambda_{ij}) + z_{ij}\lambda_{ij}^{y_{ij} + \alpha_j - 1}e^{-\lambda_{ij}(1 + \phi_i) }.
\]

Finally, we recall that we consider several approaches in Section \ref{sec:zanidm_inference} for updating $\alpha_j$, whose full conditional distribution is given in  \eqref{eq:alpha_j_full_conditional}. Two of these approaches --- namely, MH with a Gaussian random walk and the slice sampler of \citet{Neal2003} --- work by updating $\beta_j$ according to the re-parameterisation $\ln \alpha_j=\beta_j$ and the prior $\beta_j \sim \operatorname{Normal}\lbrack m_j, s_j^2\rbrack$. As these approaches are quite standard, we do not describe them further here. Instead, we provide some details on the data augmentation strategies proposed by \citet{Hamura2023}, who present a general scheme for cases where the parameter of interest appears as the argument of a gamma function, as occurs with the $\bm{\alpha}$ parameter in the Dirichlet, Dirichlet-multinomial, and indeed ZANIDM distributions. Recall that under the prior $\alpha_j \sim \operatorname{Gamma}\lbrack c_j, d_j\rbrack$,
where the category-specific hyper-parameters $c_j$ and $d_j$ are known, our target $\pi(\alpha_j)$, given by  \eqref{eq:target_alpha_da_ptn}, is not straightforward to sample from. The main idea of \citet{Hamura2023} is to introduce auxiliary variables, such that the target can be approximated
by proposing from an independent power-truncated-normal (PTN) distribution and conducting a simple MH step.
These strategies lead to a three-step process, which we adapt to the ZANIDM setting as follows below.

$\bullet$ First step$\colon$ Beta data augmentation for dealing with the term $1/\Gamma(\alpha_j)^{t_j}$ in
\eqref{eq:target_alpha_da_ptn}.\\
Consider $\rho_{kj} \sim \operatorname{Beta}\lbrack \alpha_j + (k - 1)/t_j, (t_j-k+1) / t_j \rbrack,~ k\in\{2,\ldots, t_j\}$.
Then, the target, conditional now on
the auxiliary variables $\bm{\rho}_j = (\rho_{2j}, \ldots, \rho_{t_{j}})$, is given by
\begin{align*}
\pi(\alpha_j \mid \bm{\rho}_j) &\propto
\alpha_j^{c_j + t_j - 1/2}
\exp\left\lbrack -\alpha_j \left(d_j - \sum_{i\colon \lambda_{ij}>0} \ln \lambda_{ij}
-\sum_{k=2}^{t_j} \ln \rho_{kj} - t_j  \right)\right\rbrack\times
\dfrac{1}{\alpha_j^{t_j\alpha_j}}C(\alpha_j),
\end{align*}
where $C(\alpha_j)=\dfrac{(t_j\alpha_j)^{t_j\alpha_j - 1/2}}{\Gamma(t_j\alpha_j)e^{t_j\alpha_j}}$.

$\bullet$ Second step$\colon$ Gamma data augmentation for dealing with the term $1/\alpha_j^{t_j\alpha_j}$.\\
By introducing the auxiliary variable
$w_j \sim \operatorname{Gamma}\lbrack t_j\alpha_j, t_j\alpha_j^2\rbrack$ and defining $p_j^\star = t_{j} + c_j$, $a_j^\star = t_{j} w_j$,
and $b_j^\star = t_{j}\ln w_j + 2 t_{j} + \sum_{i\colon \lambda_{ij} > 0}\ln \lambda_{ij} + \sum_{k=2}^{t_{j}} \ln \rho_{jk} - d_j$, we obtain
\begin{equation}\label{eq:target_da_ptn_rho_w}
\pi(\alpha_j \mid \bm{\rho}_j, w_j)
\propto \alpha_j^{p_j^\star - 1} \exp\left\lbrack-a_j^\star\alpha_j^2 + b_j^\star\alpha_j \right\rbrack
C(\alpha_j)^2.
\end{equation}

$\bullet$ Third step$\colon$ Metropolis-Hastings with independent PTN proposals.\\
The target in  \eqref{eq:target_da_ptn_rho_w}, now conditioned on the
auxiliary variables $\bm{\rho}_j$ and $w_j$, can be written as
$
\pi(\alpha_j \mid \bm{\rho}_j, w_j) \propto f_{\operatorname{PTN}}(\alpha_j; p_j^\star, a_j^\star, b_j^\star)C(\alpha_j)^2,
$
where $f_{\operatorname{PTN}}(x; p, a, b)$ denotes the probability density function of a PTN random variable\footnote{If $X \sim \operatorname{PTN}\lbrack p, a, b\rbrack$, then $f(x) \propto x^{p-1}e^{-ax^2 + bx}$ for $x>0$, $a>0$, $p>0$, and $b\neq0$. We note that this unnormalised density is of the same form as that of the modified-half-normal distribution \citep{Sun2023}.}. Following \citet{Hamura2023}, we consider independent PTN proposals with the same parameters, i.e.,
$\alpha_j^{(t)} \sim \operatorname{PTN}\lbrack p_j^\star, a_j^\star, b_j^\star \rbrack$, for which we use rejection sampling.
Then, due to the proposals being independent and of the same form as the target, the MH acceptance probability to move from $\alpha_j^{(t-1)}$ to $\alpha_j^{(t)}$ simplifies to $\min\{1,C(\alpha_j^{(t)})^2/C(\alpha_j^{(t-1)})^2\}$.
As shown by \citet{Hamura2023}, the factor $C(\alpha_j)$ is almost constant when $\alpha_j$ is not extremely small, and the acceptance probability is close to $1$.

The overall Bayesian inference scheme for the parameters of the ZANIDM distribution is presented in \autoref{alg:zanidm_inference}.\medskip

\begin{algorithm}[H]
\caption{MCMC inference algorithm for the ZANIDM distribution.}
\label{alg:zanidm_inference}
\SetAlgoLined
\SetAlgoItemize
\SetKwInput{KwInput}{Input}
\SetKwInput{KwOutput}{Output}
\SetKwInput{KwInit}{Initialise}
\KwInput{Data $\{\mathbf{y}_i; 1 \leq i \leq n\}$ and number of Monte Carlo iterations $R_{\mathrm{MCMC}}$.}
\KwInit{$\phi^{(0)}_i$, $\lambda^{(0)}_{ij}$, $z^{(0)}_{ij}$, $\zeta_j^{(0)}$, and $\alpha_j^{(0)}$ for $i \in\{1,\ldots, n\}$ and $j \in\{1,\ldots, d\}$.}
\For{iterations $t$ from $1$ to $R_{\mathrm{MCMC}}$}{
\For{categories $j$ from $1$ to $d$}{
\begin{enumerate}
\item Sample $(\zeta_j^{(t)} \mid \mathbf{y}, \bm{\phi}^{(t-1)},
\mathbf{z}^{(t-1)})$ from its full conditional in  \eqref{eq:posterior_zeta_zanidm}.

\item Sample either from $\pi(\alpha_j)$ in  \eqref{eq:target_alpha_da_ptn} or
from $\pi(\beta_j)$ in  \eqref{eq:target_beta_mh_ss}.

\item Update $(z_{ij}^{(t)} \mid y_{ij}, \phi_i^{(t-1)})$ from its collapsed conditional distribution in  \eqref{eq:zanidm_z_update}.

\item Update $(\lambda_{ij}^{(t)} \mid y_{ij}, z_{ij}^{(t)}, \phi_i^{(t-1)})$ from its full conditional in  \eqref{eq:zanidm_lambda_update}.
\end{enumerate}}
$\bullet$ Update $(\phi_{i}^{(t)} \mid \mathbf{y}_i, \bm{\lambda}^{(t)}_i) \sim \operatorname{Gamma}\left\lbrack N_i, \sum_{j=1}^n\lambda^{(t)}_{ij}\right\rbrack, \quad \mathrm{for} \quad i\in\{1, \ldots, n\}$.
}
\end{algorithm}
 \subsection{Moment generating functions via mixture properties}\label{app:mgf}

We can find the moment generating function (MGF) for both distributions using the moment properties of mixtures with $g(\mathbf{Y}) = e^{\mathbf{t}\cdot \mathbf{Y}}$.
By identifying the component-specific distributions based on our novel stochastic representations of the
ZANIM and ZANIDM distributions in terms of finite mixtures,
we note that the first two terms $\delta_{\mathbf{0}_d}(\cdot)$ and $N\mathbf{e}_d^{(j)}$ are degenerate random vectors, while the
remaining terms follow multinomial and DM distributions, respectively. Thus,
\begin{equation}\label{eq:zanim_mgf_deriv}
M_{\mathbf{Y}}(\mathbf{t}) = \mathbb{E}\lbrack e^{\mathbf{t}\cdot \mathbf{Y}} \rbrack =
\eta_0M_{\delta_{\mathbf{0}_d}(\cdot)}(\mathbf{t}) +
\sum_{j=1}^d\eta_N^{(j)}M_{N\mathbf{e}_d^{(j)}}(\mathbf{t}) +
\eta_dM_{\mathbf{X}}(\mathbf{t}) +
\sum_{\mathcal{K} \in \mathfrak{K}}\eta_\mathcal{K}M_{\mathbf{X}^\mathcal{K}}(\mathbf{t}),
\end{equation}
where $\mathbf{t} = (t_1, \ldots, t_d)$, and the random vectors $\mathbf{X}$ and $\mathbf{X}^\mathcal{K}$ have
multinomial
or DM distributions with
appropriate dimension and parameters. Expanding the sum in  \eqref{eq:zanim_mgf_deriv} with the
MGFs of the corresponding components
trivially yields the ZANIM MGF as follows
\[
M_{\mathbf{Y}}(\mathbf{t})
=
\eta_0 + \sum_{j=1}^d\eta_N^{(j)} \exp\,(Nt_j)
+ \eta_d\left(\sum_{j=1}^d\theta_j\exp\,(t_j)\right)^N
+ \sum_{\mathcal{K} \in \mathfrak{K}}
\eta_\mathcal{K}\left(
\sum_{j\notin\mathcal{K}}\theta_j^\mathcal{K}\exp\,(t_j)
\right)^N.
\]
By way of verification, we have that the first partial derivative of $M_{\mathbf{Y}}(\mathbf{t})$ w.r.t. $t_j$ is
\begin{align*}
\dfrac{\partial M_\mathbf{Y}(\mathbf{t})}{\partial t_j} &=
N\eta_N^{(j)}\exp\,(Nt_j) +
N\eta_d\theta_j\exp\,(t_j) \left(\sum_{j=1}^d\theta_j\exp\,(t_j)\right)^{N-1} + N
\sum_{\mathcal{S}_j \in \mathfrak{S}_j} \eta_{\mathcal{S}_j}\theta_j^{\mathcal{S}_j}\exp\,(t_j)
\left\lbrack \sum_{j \notin \mathcal{S}_j}\theta_j^{\mathcal{S}_j}\exp\,(t_j) \right\rbrack^{N-1},
\end{align*}
where the partial derivative of the last term w.r.t. $t_j$ vanishes for sets outside the defined $\mathfrak{S}_j$. As both $\sum_{j=1}^d\theta_j = 1$ and $\sum_{j \notin \mathcal{S}_j}\theta_j^{\mathcal{S}_j}=1$,
it is trivial to show that the expression derived for $\mathbb{E}\lbrack Y_j \rbrack$ in \autoref{prop:zanim_moments} is recovered by setting $t_j=0$ here. Using the usual argumentation for MGFs also recovers $\mathrm{Var}\lbrack Y_j \rbrack$ under ZANIM and yields similar results for ZANIDM.
 \subsection{Posterior summaries for alternative ZANIDM inference schemes}\label{app:posterior_summaries_zanidm}

In Section \ref{sec:simstudy2}, the DA-PTN approach was used to infer the concentration parameters $\bm{\alpha}$ when fitting the ZANIDM model to data sets containing $500$ observations generated from the ZANIM and ZANIDM distributions. For completeness, we report equivalent results under the slice sampling (SS) and MH-RW approaches, along with the \texttt{ZIDM} \textsf{R} package of \citet{Koslovsky2023} in  \autoref{tab:posterior_summaries_zanidm_dgp_zanim} (for data generated from the ZANIM distribution) and \autoref{tab:posterior_summaries_zanidm_dgp_zanidm} (for data generated from the ZANIDM distribution). The performance of each approach is broadly in line with the insights gleaned from the comparative simulations in Section \ref{sec:simstudy1}. The DA-PTN results shown here are exact reproductions of the corresponding rows of \autoref{tab:posterior_summaries_sim_2}.

\begin{table}[H]
\centering
\captionsetup{width=.6\textwidth}
\caption{
Posterior means, lower (LCI) and upper (UCI) limits of $95\%$ credible intervals, and effective sample size (ESS) ratios for the parameters of the ZANIDM distribution.
We report the posterior summaries of four fits which use different sampling schemes to infer $\bm{\alpha}$.
The data are generated from the ZANIM distribution with a sample size of $500$,
using the following true parameter values: $\bm{\theta} \in \{0.05, 0.70, 0.25\}$, $\bm{\zeta} \in \{0.05, 0.15, 0.10\}$, and $N = 30$.}
\label{tab:posterior_summaries_zanidm_dgp_zanim}
\vskip-0.3cm\rule{.6\textwidth}{0.4pt}\smallskip

\begin{tabular*}{.6\textwidth}{@{\hspace{0.01\textwidth}\extracolsep{\fill}}lcrrrrr@{\hspace{0.01\textwidth}}}
Method & Parameter & Mean & $95\%$ LCI & $95\%$ UCI &  ESS ratio \\
\midrule
\multirow{6}{*}{DA-PTN} & $\alpha_1$ & 3.859 & 1.456 & 13.026 & 0.054 \\
   & $\alpha_2$ & 56.607 & 18.809 & 219.522 & 0.055 \\
   & $\alpha_3$ & 19.734 & 6.735 & 72.113 & 0.054 \\
   & $\zeta_1$ & 0.011 & 0.000 & 0.050 & 0.581 \\
   & $\zeta_2$ & 0.140 & 0.112 & 0.171 & 0.933 \\
   & $\zeta_3$ & 0.120 & 0.094 & 0.151 & 0.865 \\
   \hline
\multirow{6}{*}{SS} & $\alpha_1$ & 15.959 & 9.438 & 22.193 & 0.008 \\
   & $\alpha_2$ & 239.714 & 142.033 & 327.461 & 0.008 \\
   & $\alpha_3$ & 83.471 & 49.378 & 114.833 & 0.008 \\
   & $\zeta_1$ & 0.016 & 0.001 & 0.047 & 1.034 \\
   & $\zeta_2$ & 0.140 & 0.112 & 0.173 & 1.144 \\
   & $\zeta_3$ & 0.121 & 0.094 & 0.149 & 0.994 \\
   \hline
\multirow{6}{*}{MH-RW} & $\alpha_1$ & 7.286 & 5.646 & 9.739 & 0.009 \\
   & $\alpha_2$ & 108.572 & 84.138 & 143.077 & 0.008 \\
   & $\alpha_3$ & 37.811 & 29.346 & 49.899 & 0.008 \\
   & $\zeta_1$ & 0.014 & 0.000 & 0.042 & 1.047 \\
   & $\zeta_2$ & 0.139 & 0.110 & 0.171 & 1.005 \\
   & $\zeta_3$ & 0.122 & 0.094 & 0.152 & 1.084 \\
   \hline
\multirow{6}{*}{ZIDM} & $\alpha_1$ & 4.786 & 2.809 & 7.157 & 0.001 \\
   & $\alpha_2$ & 70.084 & 40.726 & 104.699 & 0.001 \\
   & $\alpha_3$ & 24.492 & 14.390 & 36.863 & 0.001 \\
   & $\zeta_1$ & 0.019 & 0.004 & 0.048 & 0.466 \\
   & $\zeta_2$ & 0.139 & 0.108 & 0.171 & 1.019 \\
   & $\zeta_3$ & 0.121 & 0.093 & 0.152 & 1.103 \\
\end{tabular*}
\rule{.6\textwidth}{0.4pt}
\end{table}

In \autoref{tab:posterior_summaries_zanidm_dgp_zanim}, DA-PTN, MH-RW, slice sampling, and \texttt{ZIDM} perform similarly in terms of parameter recovery, although only DA-PTN has credible intervals which contain the true values of $\bm{\zeta}$ in each case. As per Section \ref{sec:simstudy2}, inference for $\bm{\alpha}$ is poor, under all approaches, in this scenario with ZANIM as the data-generating process. In \autoref{tab:posterior_summaries_zanidm_dgp_zanidm}, where ZANIDM is the data-generating process, DA-PTN, MH-RW, and slice sampling again perform similarly, though \texttt{ZIDM} is now notably worse. Only DA-PTN has credible intervals which contain the true values of all parameters, and the ESS ratios for the $\bm{\alpha}$ parameters under \texttt{ZIDM} are unacceptably low. We conjecture that this is attributable to the joint update of $\lambda_{ij}$ and $z_{ij}$ performed by \texttt{ZIDM}.

\begin{table}[H]
\centering
\captionsetup{width=.6\textwidth}
\caption{
Posterior means, lower (LCI) and upper (UCI) limits of $95\%$ credible intervals, and effective sample size (ESS) ratios for the parameters of ZANIDM distribution.
We report the posterior summaries using different sampling schemes to infer $\bm{\alpha}$.
The data are generated from the ZANIDM distribution with a sample size of $500$,
using the following true parameter values: $\bm{\alpha} \in \{2, 28, 10\}$, $\bm{\zeta} \in \{0.05, 0.15, 0.10\}$, and $N = 30$.
}
\label{tab:posterior_summaries_zanidm_dgp_zanidm}
\vskip-0.3cm\rule{.6\textwidth}{0.4pt}\smallskip

\begin{tabular*}{.6\textwidth}{@{\hspace{0.01\textwidth}\extracolsep{\fill}}lcrrrrr@{\hspace{0.01\textwidth}}}
Method & Parameter & Mean & $95\%$ LCI & $95\%$ UCI &  ESS ratio \\
 \midrule
  \multirow{6}{*}{DA-PTN} & $\alpha_1$ & 1.241 & 0.787 & 2.301 & 0.148 \\
   & $\alpha_2$ & 18.822 & 11.284 & 35.420 & 0.145 \\
   & $\alpha_3$ & 6.829 & 4.130 & 12.985 & 0.133 \\
   & $\zeta_1$ & 0.025 & 0.001 & 0.095 & 0.476 \\
   & $\zeta_2$ & 0.129 & 0.101 & 0.160 & 0.902 \\
   & $\zeta_3$ & 0.093 & 0.068 & 0.120 & 0.827 \\
   \hline
\multirow{6}{*}{SS} & $\alpha_1$ & 1.497 & 1.191 & 1.880 & 0.186 \\
   & $\alpha_2$ & 23.389 & 19.012 & 28.844 & 0.151 \\
   & $\alpha_3$ & 8.486 & 6.890 & 10.470 & 0.153 \\
   & $\zeta_1$ & 0.032 & 0.001 & 0.093 & 0.907 \\
   & $\zeta_2$ & 0.128 & 0.100 & 0.158 & 0.863 \\
   & $\zeta_3$ & 0.094 & 0.070 & 0.120 & 1.052 \\
   \hline
\multirow{6}{*}{MH-RW} & $\alpha_1$ & 1.499 & 1.173 & 1.933 & 0.034 \\
   & $\alpha_2$ & 23.616 & 18.841 & 29.927 & 0.025 \\
   & $\alpha_3$ & 8.566 & 6.805 & 10.934 & 0.030 \\
   & $\zeta_1$ & 0.031 & 0.001 & 0.092 & 0.350 \\
   & $\zeta_2$ & 0.127 & 0.100 & 0.157 & 1.050 \\
   & $\zeta_3$ & 0.094 & 0.070 & 0.122 & 1.012 \\
   \hline
\multirow{6}{*}{ZIDM} & $\alpha_1$ & 1.481 & 1.204 & 1.770 & 0.006 \\
   & $\alpha_2$ & 22.933 & 18.851 & 26.058 & 0.006 \\
   & $\alpha_3$ & 8.336 & 6.879 & 9.574 & 0.006 \\
   & $\zeta_1$ & 0.041 & 0.007 & 0.094 & 0.556 \\
   & $\zeta_2$ & 0.128 & 0.100 & 0.159 & 0.961 \\
   & $\zeta_3$ & 0.094 & 0.069 & 0.122 & 1.110 \\
\end{tabular*}
\rule{.6\textwidth}{0.4pt}
\end{table}
 \subsection{Additional simulation results with balanced parameter settings}\label{app:balanced_experiments}

The simulation design in Section \ref{sec:simstudy2} was particularly challenging by virtue of matching the data-generating processes to the parameter settings used in \autoref{fig:comparison_zanim_zanidm_pmf}, in the sense that the ZANIM parameters $\bm{\theta}$ and ZANIDM parameters $\bm{\alpha}$ were heavily imbalanced. For completeness, we conduct additional simulation experiments with data sets containing $500$ observations generated from both distributions using balanced values for the $\bm{\theta}$ and $\bm{\alpha}$ parameters of the
ZANIM and ZANIDM distributions, respectively. Specifically, we keep the same number of $d=3$ categories, the same number of trials $N=30$, and the same $\bm{\zeta}=(0.05, 0.15, 0.10)$ configuration for the zero-inflation parameters in each case, with $\bm{\theta}=(1/3, 1/3, 1/3)$ and $\bm{\alpha}=(1.0, 1.0, 1.0)$ under the respective DGPs. It is important to stress that the total concentration $\alpha_s=\sum_{j=1}^d\alpha_j=3$ is quite low. As per Section \ref{sec:simstudy2}, we consider only the DA-PTN approach to infer the $\bm{\alpha}$ parameters when fitting the ZANIDM model.

The posterior summaries in \autoref{tab:balanced_zanim} show that the true values of the $\bm{\zeta}$ parameters are within the $95\%$ credible intervals throughout, with the exception of $\zeta_2$ for the ZANIM model fitted to data generated from the ZANIDM distribution. Furthermore, inference for the $\bm{\theta}$ parameters under the ZANIM model and the $\bm{\alpha}$ parameters under the ZANIDM model are satisfactory when the DGP matches the model. Notably, the ESS ratios for the $\bm{\alpha}$ parameters of ZANIDM are much improved in these balanced cases, particularly when the data are generated from the ZANIDM distribution, compared to the corresponding values in \autoref{tab:posterior_summaries_sim_2}. It is also notable that the posterior mean estimates of the $\bm{\theta}$ parameters when the ZANIM model is fitted to data generated from the ZANIDM distribution are all approximately $1/3$.
\begin{table}[H]
\centering
\captionsetup{width=\textwidth}
\caption{
Posterior means, lower (LCI) and upper (UCI) limits of $95\%$ credible intervals, and effective sample size (ESS) ratios for the parameters of the ZANIM and ZANIDM models.
We report the posterior summaries for each model under two data-generating processes (DGPs), which are based on the ZANIM and ZANIDM distributions. For each DGP, $500$ samples are simulated from the corresponding distribution using balanced parameter configurations.}
\label{tab:balanced_zanim}
\vskip-0.3cm\hrule\smallskip
\begin{tabular*}{\textwidth}{@{\hspace{0.01\textwidth}\extracolsep{\fill}}llcrrrrr@{\hspace{0.01\textwidth}}}
DGP & Model & Parameter & Mean & $95\%$ LCI & $95\%$ UCI &  ESS ratio \\
  \midrule
\multirow{12}{*}{\shortstack[l]{ZANIM:\\$\bm{\theta} \in \{0.05, 0.70, 0.25\}$,\\$\bm{\zeta} \in \{0.05, 0.15, 0.10\}$,\\$N=30$ trials.}} & \multirow{6}{*}{ZANIDM} & $\alpha_1$ & 9.590 & 6.773 & 14.080 & 0.276 \\
   &  & $\alpha_2$ & 9.633 & 6.759 & 13.932 & 0.295 \\
   &  & $\alpha_3$ & 10.086 & 7.139 & 14.842 & 0.267 \\
   &  & $\zeta_1$ & 0.036 & 0.021 & 0.056 & 0.857 \\
   &  & $\zeta_2$ & 0.140 & 0.109 & 0.170 & 0.967 \\
   &  & $\zeta_3$ & 0.121 & 0.095 & 0.151 & 0.846 \\
   \cmidrule{2-8} & \multirow{6}{*}{ZANIM} & $\theta_1$ & 0.328 & 0.321 & 0.335 & 0.864 \\
   &  & $\theta_2$ & 0.333 & 0.326 & 0.341 & 0.969 \\
   &  & $\theta_3$ & 0.339 & 0.331 & 0.346 & 1.005 \\
   &  & $\zeta_1$ & 0.036 & 0.022 & 0.054 & 1.010 \\
   &  & $\zeta_2$ & 0.140 & 0.111 & 0.173 & 0.937 \\
   &  & $\zeta_3$ & 0.120 & 0.094 & 0.149 & 0.947 \\
   \hline
\multirow{12}{*}{\shortstack[l]{ZANIDM:\\$\bm{\alpha}\in \{2.0, 28.0, 10.0\}$,\\$\bm{\zeta}\in \{0.05, 0.15, 0.10\}$,\\$N=30$ trials.}} & \multirow{6}{*}{ZANIDM} & $\alpha_1$ & 0.885 & 0.718 & 1.068 & 0.537 \\
   &  & $\alpha_2$ & 0.992 & 0.790 & 1.227 & 0.528 \\
   &  & $\alpha_3$ & 0.966 & 0.784 & 1.193 & 0.522 \\
   &  & $\zeta_1$ & 0.038 & 0.004 & 0.076 & 0.718 \\
   &  & $\zeta_2$ & 0.182 & 0.134 & 0.224 & 0.923 \\
   &  & $\zeta_3$ & 0.119 & 0.081 & 0.158 & 0.882 \\
   \cmidrule{2-8}
 & \multirow{6}{*}{ZANIM} & $\theta_1$ & 0.312 & 0.305 & 0.321 & 0.975 \\
   &  & $\theta_2$ & 0.352 & 0.343 & 0.360 & 1.032 \\
   &  & $\theta_3$ & 0.336 & 0.327 & 0.345 & 1.123 \\
   &  & $\zeta_1$ & 0.108 & 0.082 & 0.136 & 1.024 \\
   &  & $\zeta_2$ & 0.227 & 0.192 & 0.264 & 1.027 \\
   &  & $\zeta_3$ & 0.171 & 0.140 & 0.206 & 1.050 \\
\end{tabular*}
\hrule
\end{table}
Finally, \autoref{tab:elpd_balanced} gives the ELPD results for both models under both DGPs. As per \autoref{tab:elpd_sim2}, the ELPD favours the distribution used to generate the data. As regards the DM model included in this comparison, we note that it outperforms the ZANIM model under the ZANIDM DGP. This was not the case in \autoref{tab:elpd_sim2}, which is likely due to the similarity of the $\mathrm{DI}\lbrack Y_j \rbrack$ indices of both distributions under the parameter settings used in Section \ref{sec:simstudy2} (see \autoref{tab:theoretical_moments}). Under the balanced parameter settings used here to generate the data, these indices are much higher under ZANIDM than ZANIM, which indicates that the $\bm{\alpha}$ parameters contribute more to the overdispersion in the data than the $\bm{\zeta}$ parameters.
\begin{table}[H]
\centering
\captionsetup{width=.6\textwidth}
\caption{
Bayesian model evaluation metrics for different models with data simulated under two data-generating processes (DGPs) based on the ZANIM and ZANIDM distributions.
We report the expected log-predictive density $(\widehat{\operatorname{elpd}})$ and its standard error ($\operatorname{se}(\widehat{\operatorname{elpd}})$).
For each DGP, $500$ samples are simulated from the corresponding distribution using balanced parameter configurations.}
\label{tab:elpd_balanced}
\vskip-0.3cm\rule{.6\textwidth}{0.4pt}\smallskip

\begin{tabular*}{.6\textwidth}{@{\hspace{0.01\textwidth}\extracolsep{\fill}}llrr@{\hspace{0.01\textwidth}}}
DGP & Model & $\widehat{\operatorname{elpd}}$ & $\operatorname{se}(\widehat{\operatorname{elpd}})$
\\
 \midrule
\multirow{4}{*}{ZANIM} & ZANIM & $-2469.693$ & 22.395 \\
   & DA-PTN & $-2565.440$ & 11.877 \\
   & DM & $-3018.176$ & 18.528 \\
   & Multinomial & $-3809.590$ & 114.270 \\
\hline
  \multirow{4}{*}{ZANIDM} & DA-PTN & $-2986.729$ & 18.857 \\
   & DM & $-3011.296$ & 17.112 \\
   & ZANIM & $-4800.192$ & 112.813 \\
   & Multinomial & $-7481.003$ & 179.037 \\
\end{tabular*}
\rule{.6\textwidth}{0.4pt}
\end{table} \subsection{Additional results for the microbiome data analyses}\label{app:add_microbiome}

In Section \ref{sec:microbiome}, both the ZANIM and ZANIDM models were applied to a real human gut microbiome dataset from \citet{Wu2011}. \autoref{fig:posterior_chosen_taxa} showed trace plots for certain parameters of interest for the \textit{Actinomycinae} and \textit{Prevotella} taxa.
We now provide complementary results for the remaining taxa.
\autoref{fig:trace_plot_zetas} shows the trace plots of the $\zeta_j$ parameters for all taxa under both the ZANIM and ZANIDM models, while \autoref{fig:trace_plot_thetas} and \autoref{fig:trace_plot_alphas} show the
trace plots of all $\theta_j$ and $\alpha_j$ parameters, under the ZANIM and ZANIDM models, respectively.
\begin{figure}[H]
    \centering
    \includegraphics[width=1.0\linewidth]{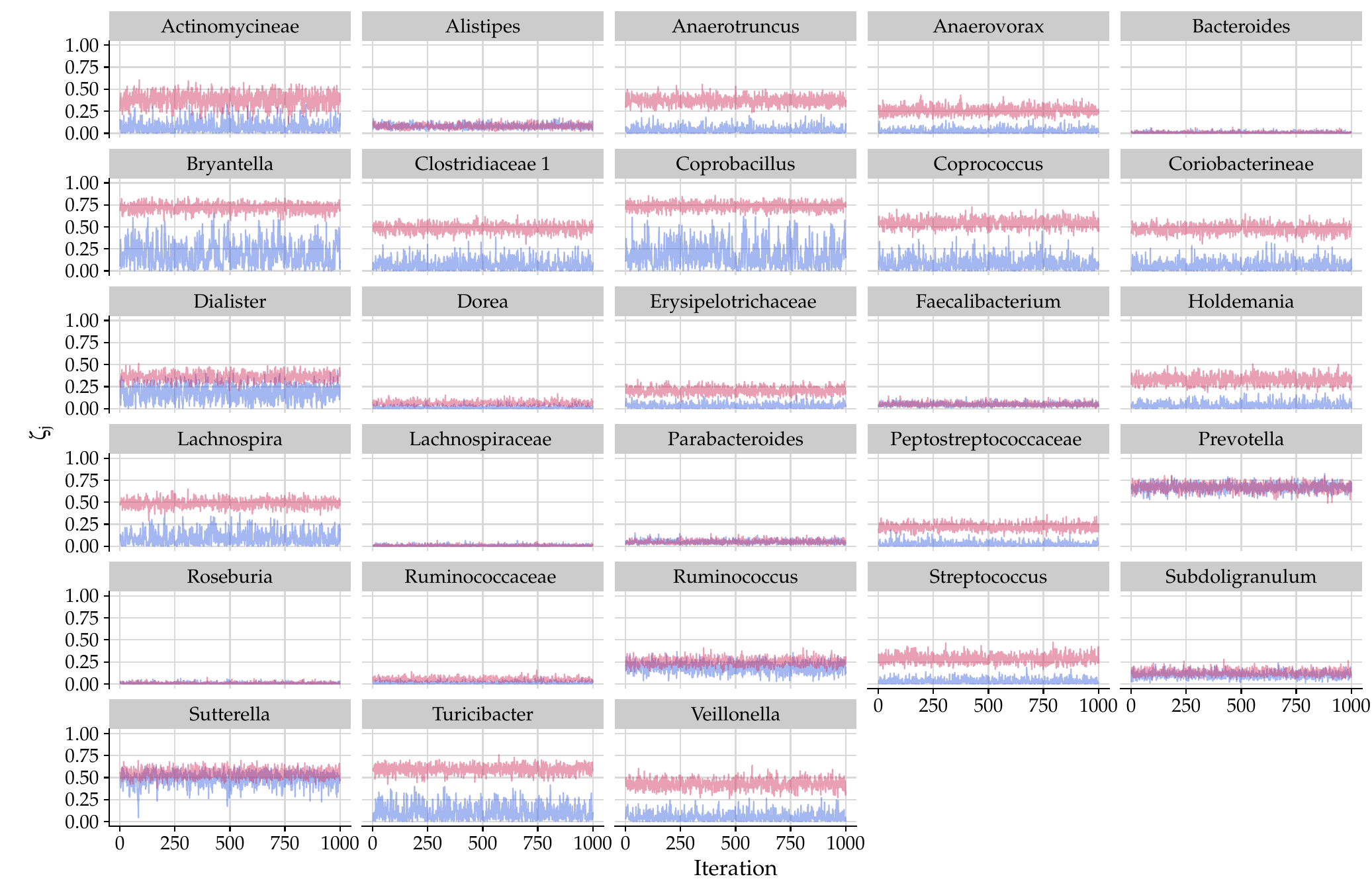}
    \caption{Trace plots of the posterior draws of $\bm{\zeta}$  for all taxa under ZANIM (red) and ZANIDM (blue) after burn-in and thinning.}
    \label{fig:trace_plot_zetas}
\end{figure}
\begin{figure}[H]
    \centering
    \includegraphics[width=1.0\linewidth]{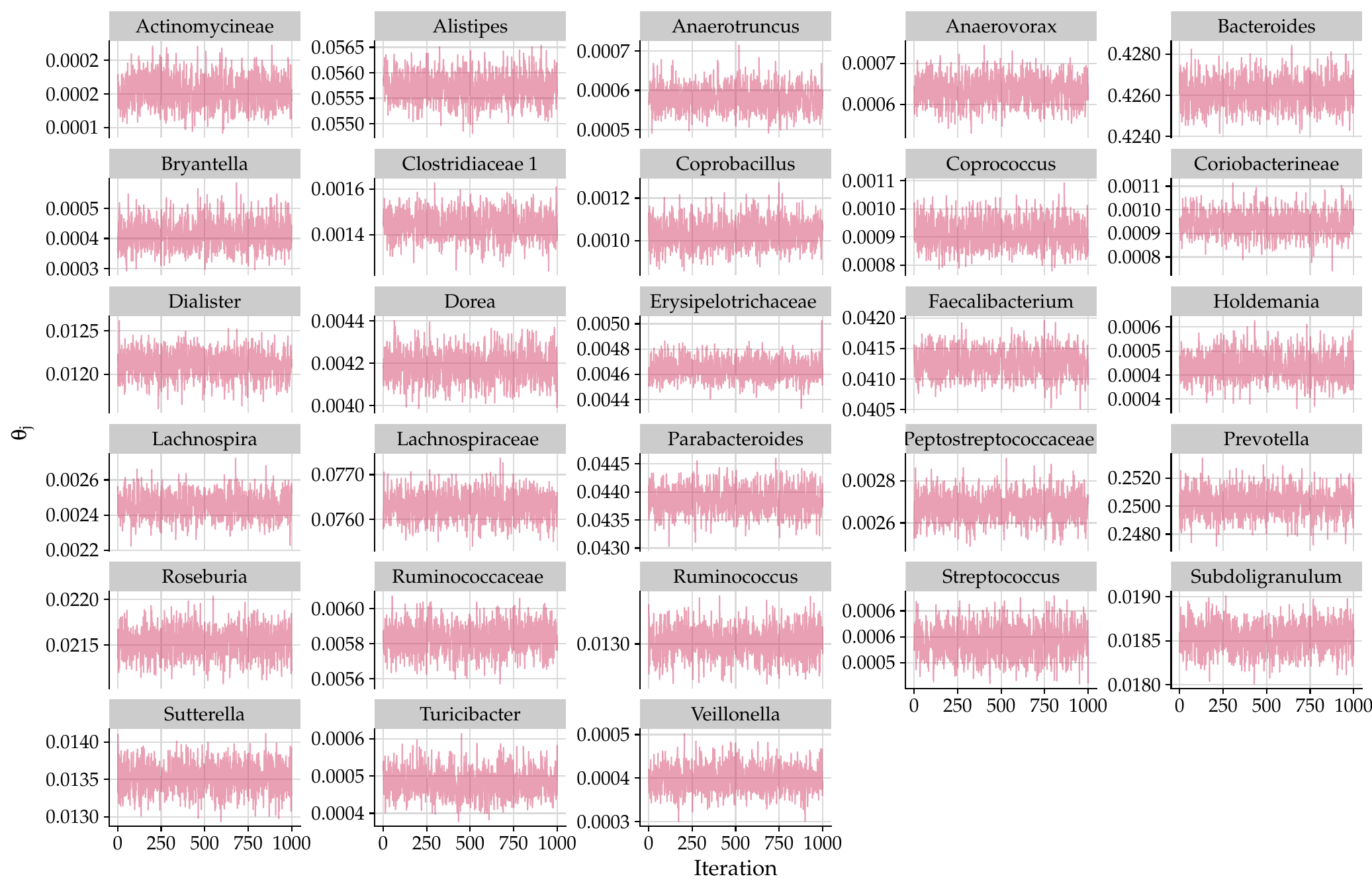}
    \caption{Trace plots of the posterior draws of $\bm{\theta}$ for all taxa under ZANIM after burn-in and thinning.}
    \label{fig:trace_plot_thetas}
\end{figure}
\begin{figure}[H]
    \centering
    \includegraphics[width=1.0\linewidth]{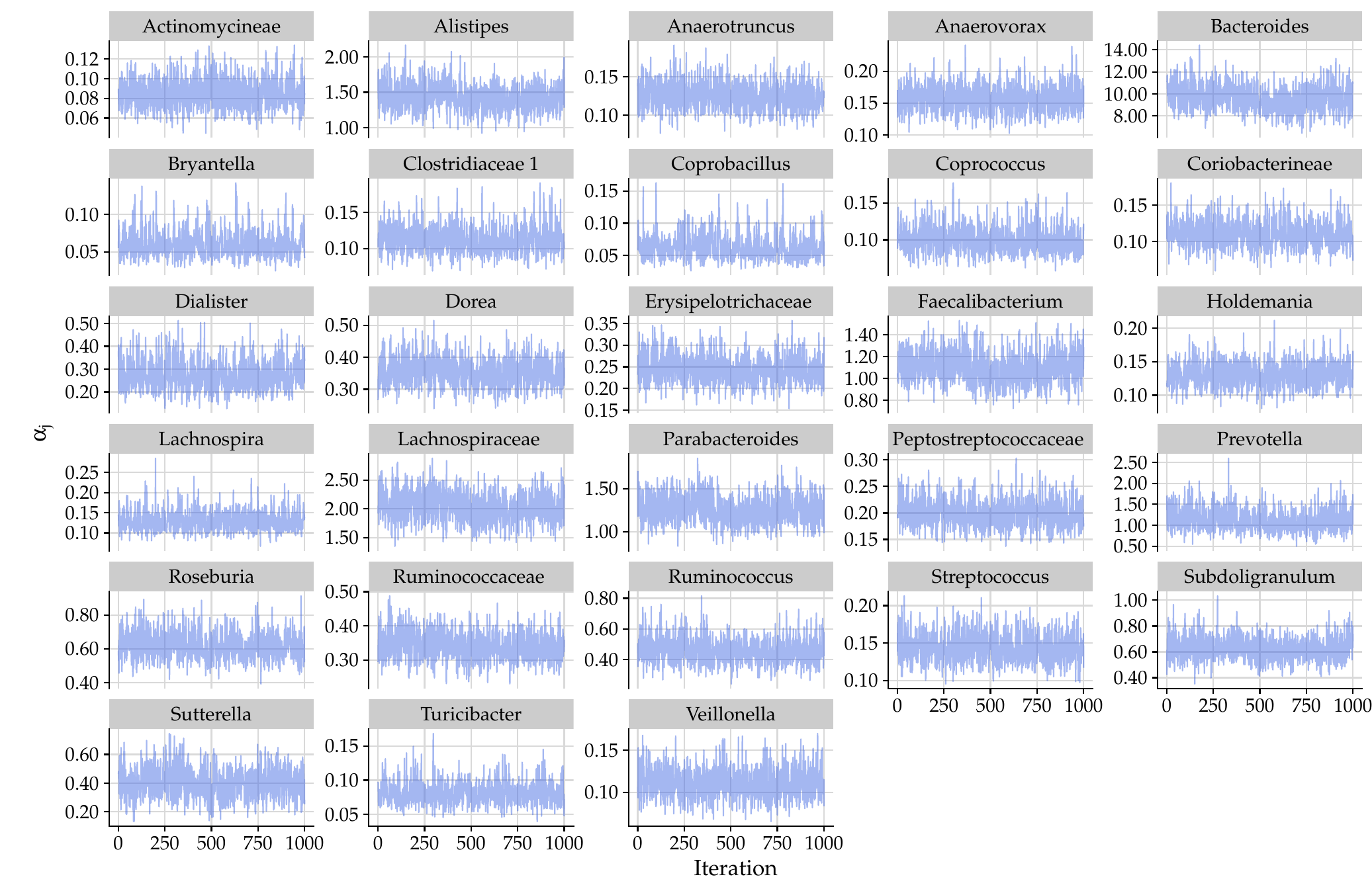}
    \caption{Trace plots of the posterior draws of $\bm{\alpha}$ for all taxa under ZANIDM after burn-in and thinning.}
    \label{fig:trace_plot_alphas}
\end{figure}

We conclude the analyses of the microbiome data by reporting
the posterior means of the relative abundances.
We denote the relative abundances by $\varphi_{ij}$ for all $i\in\{1, \ldots, 98\}$
observations and $j\in\{1, \ldots, 28\}$ taxa.
Under the ZANIDM model, these quantities are easily obtained by
normalising the posterior draws of the latent parameter $\lambda^{(m)}_{ij}$, i.e,
$\varphi^{(m)}_{ij} = \lambda^{(m)}_{ij} / \sum_{k=1}^d\lambda^{(m)}_{ik}$, where
$m\in\{1\ldots, M\}$ indexes the number of valid posterior samples.
Similar relative abundance estimates can also be derived under the ZANIM model using the $\lambda_j$ parameters and the Bernoulli latent variables $z_{ij}$ via $\varphi_{ij}^{(m)} = z^{(m)}_{ij}\lambda^{(m)}_{j} / \sum_{k=1}^d z^{(m)}_{ik}\lambda^{(m)}_{k}$.
We refer the reader to Section \ref{sec:inference} to recall the details of the parameters and latent variables employed in the inference schemes for the ZANIM and ZANIDM models.

Following \citet{Koslovsky2023}, \autoref{fig:posterior_abundance} presents heat maps of the posterior means of
$\varphi_{ij}$ under both models. In each case, \textit{Bacteroides} appears to be the most abundant taxon overall. However, this pattern is not consistent across
individuals, with some showing low to moderate levels, particularly under the ZANIDM model. In any case, several other taxa exhibit notable variation across individuals and these individual-level differences could be further investigated in future work through the incorporation of covariates to better
explain the observed variability.
\begin{figure}[H]
    \centering
    \includegraphics[width=1.0\linewidth]{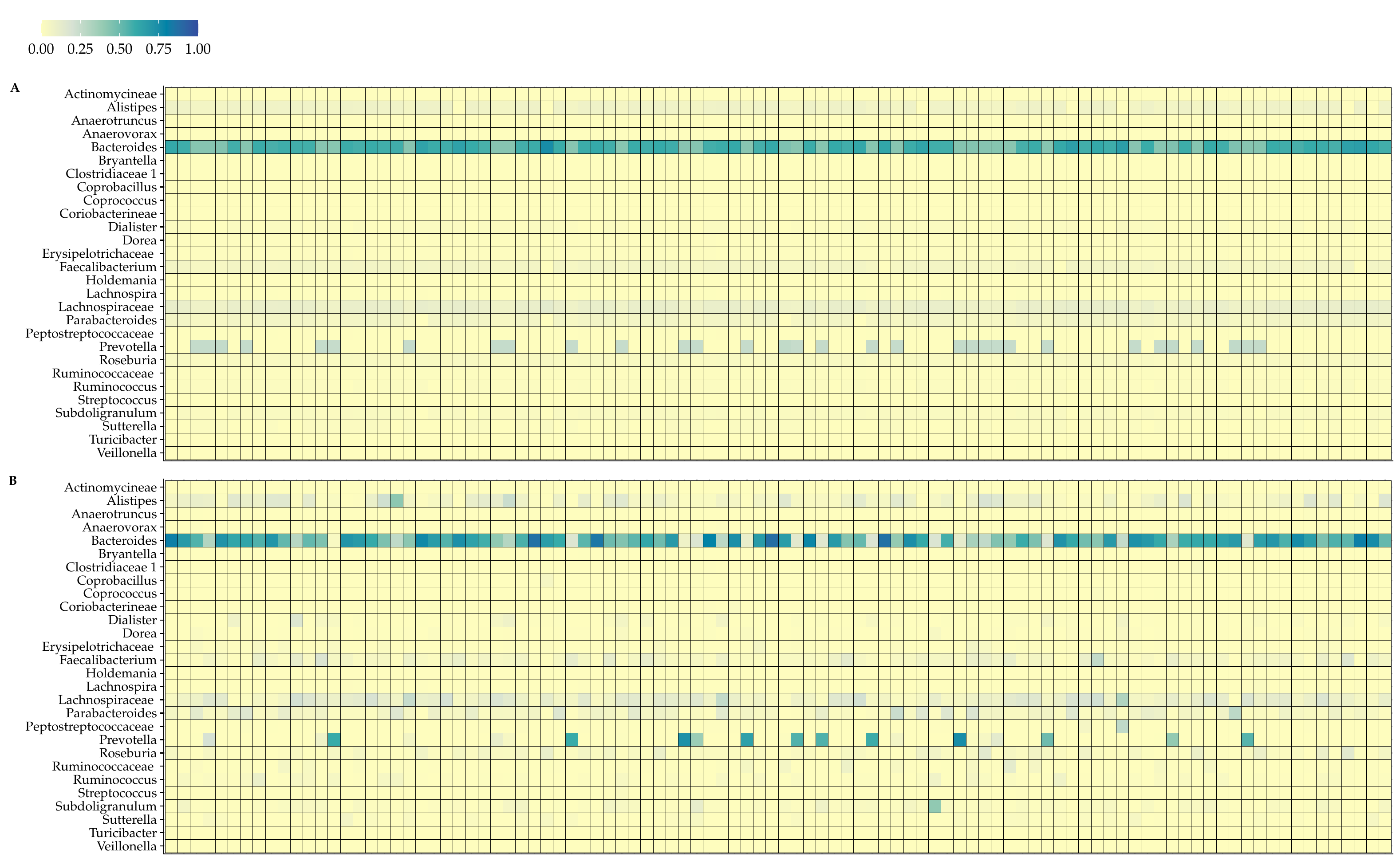}
    \caption{Heat maps of the posterior means of the relative abundances. The rows are the taxa and the columns are the individual units. \textbf{A}: ZANIM. \textbf{B}: ZANIDM.}
    \label{fig:posterior_abundance}
\end{figure}
\end{document}